\documentclass[18pt,english]{article}
\usepackage[utf8]{inputenc}
\usepackage[english]{babel}
\usepackage[a4paper, total={6in, 9in}]{geometry}
\usepackage[pdftex]{graphicx}
\usepackage[table]{xcolor}
\usepackage{amsfonts,amsthm,amsmath,amssymb,latexsym,comment}
\usepackage{float}
\usepackage{amssymb}
\usepackage{amsmath}
\usepackage{amsthm}
\usepackage{eurosym}
\usepackage{hyperref}
\usepackage{hhline}
\usepackage{chicagor}
\newcommand{\commentout}[1]{}
\newtheorem{proposition}{Proposition}
\newtheorem{corollary}{Corollary}
\newtheorem{theorem}{Theorem}

\newtheorem{lemma}{Lemma}
\newtheorem{definition}{Definition}

\title{Lower Bounds Implementing Mediators
  in Asynchronous Systems\thanks{Authors supported in part 
by MURI (MultiUniversity Research
Initiative) under grant 
W911NF-19-1-0217, by the ARO under grant W911NF-17-1-0592,
by  the NSF under grants IIS-1703846 and IIS-1718108, and by a grant
from the Open Philosophy Foundation.}}
\author{Ivan Geffner \\ Cornell University \\ ieg8@cornell.edu \and
  Joseph Y. Halpern \\ Cornell University \\ halpern@cs.cornell.edu}
  \date{}

\begin{document}

\maketitle

\begin{abstract}
Abraham, Dolev, Geffner, and Halpern \citeyear{ADGH19} proved that, in
  asynchronous systems, a \emph{$(k,t)$-robust equilibrium} for $n$
players and a trusted mediator can be implemented without the mediator
as long as $n > 4(k+t)$, where an equilibrium is $(k,t)$-robust if,
roughly speaking, no coalition of $t$ players can decrease the payoff
of any of the other players, and no coalition of $k$ players can
increase their payoff by deviating. We prove that this bound is tight,
in the sense that if $n \le 4(k+t)$ there exist $(k,t)$-robust
equilibria with a mediator that cannot be implemented by the players
alone. Even though implementing $(k,t)$-robust mediators seems closely
related to implementing asynchronous multiparty $(k+t)$-secure
computation \cite{BCG93}, to the best of our knowledge there is
no known straightforward reduction from one problem to
another. Nevertheless, we show that there is a non-trivial reduction
from a slightly weaker notion of $(k+t)$-secure computation, which we
call \emph{$(k+t)$-strict secure computation}, to implementing
$(k,t)$-robust mediators. We prove the desired lower bound by showing
that there are functions on $n$ variables that cannot be
$(k+t)$-strictly securely computed if $n \le 4(k+t)$. This also
provides a simple alternative proof for the well-known lower bound of
$4t+1$ on asynchronous secure computation in the presence of up to $t$
malicious agents \cite{ADS20,BKR94,canetti96studies}.
\end{abstract}


\maketitle

\section{Introduction}
Ben-Or, Goldwasser, and Wigderson \citeyear{BGW88} (BGW from now on) showed
that given a finite domain $D$, a function $f: D^n \rightarrow D$
can be \emph{$t$-securely computed} by $n$ agents in a synchronous network with
private authenticated channels as long as $n > 3t$, where $t$ is a
bound on the number of malicious players.  
Roughly speaking, ``$t$-securely computed'' means that all honest
agents
 correctly compute the output of $f$,
while a group of up to $t$ malicious 
agents 
 can learn
nothing about the players' inputs beyond what can be learned the output of $f$.
%
Ben-Or, Canetti, and Goldreich \citeyear{BCG93} (BCG from now on) later
provided analogous results for the asynchronous case: a function $f:
D^n \rightarrow D$ can be 
$t$-securely 
computed by $n$ agents if 
$n > 4t$.

Abraham, Dolev, Gonen, and Halpern \citeyear{ADGH06} consider a
problem related to secure function computation that has deep roots in
the game-theory literature.  
The agents in this case have an input and play a game $\Gamma$.
They make a move in the game and get a
payoff that depends on the action profile (i.e., the move made by 
each agent).  Of course, if the moves consist of outputting a value in
$D$, then we can view function computation as a game, where the agents
(or players\footnote{We typically use the term \emph{agent} when there is no underlying game and \emph{player} when there is.}) 
get a payoff depending on the value that they output.

Secure function computation is often viewed as a game with a trusted
third party, or \emph{mediator}.  Roughly speaking, we want the
outcome to be the same as if the agents had sent their input values
$\vec{x}$ to the mediator, who then sends back $f(\vec{x})$.  
Motivated by this viewpoint, Abraham et al.~considered two 
extensions
of a \emph{Bayesian game}  $\Gamma$ (a game where the agents each have
an input, or \emph{type}, make a single move, and then get a payoff
that depends on the type profile---the tuple consisting of each
agent's type---and move profile).
\commentout{ 
  The first is a game
$\Gamma_d$ with a trusted mediator, where, after communicating with
the mediator, 
the players make a move in the underlying game $\Gamma$ 
and get payoffs as specified by $\Gamma$.  The second is a
communication game denoted $\Gamma_{ACT}$, where there is no mediator,
but the players can communicate with each other, and then make a move
and get payoffs as in the underlying game $\Gamma$. 
}
The first is a game
$\Gamma_d$ with a trusted mediator, where, after a communication phase
in which the players can communicate with the mediator, they make a
move in the underlying game $\Gamma$ and get 
the same payoffs as they would in $\Gamma$. The second is a
\emph{communication game}, denoted 
$\Gamma_{ACT}$, where there is no mediator, but players can
communicate with each other before making a move in the underlying game.

Combining ideas
 from game theory and distributed computing, Abraham et al.~were interested in
 what they called \emph{$(k,t)$-robust equilibria}.  
These are strategy profiles (i.e., a strategy for each agent)  where,
roughly speaking, no coalition
of $t$ players can decrease the payoff of any of the other players,
and no coalition of $k$ players can increase their payoff by
deviating.  They showed, among other things, that if $n > 3(k+t)$ 
and there exists a $(k,t)$-robust 
equilibrium  $\vec{\sigma} + \sigma_d$ in the mediator game $\Gamma_d$ (where
$\vec{\sigma}$ represents the players' strategies and $\sigma_d$
represents the mediator's strategy), then there exists a
$(k,t)$-robust equilibrium $\vec{\sigma}_{ACT}$ in $\Gamma_{ACT}$ in a
synchronous setting 
such that, for all inputs, $\vec{\sigma}_{ACT}$ and $\vec{\sigma} + \sigma_d$
produce the same distribution over outputs if no player deviates.
\commentout{
The case where $k=0$ can be viewed as a generalization of the BGW
result. 
}
They also proved a matching lower bound \cite{ADH07}.

Abraham, Dolev, Geffner, and Halpern \citeyear{ADGH19} (ADGH from now
on) extended this result to the asynchronous setting.  
They showed that if $n > 4(k+t)$ and there exists a $(k,t)$-robust 
equilibrium  $\vec{\sigma} + \sigma_d$ in the mediator game $\Gamma_d$,
then there exists a
$(k,t)$-robust equilibrium $\vec{\sigma}_{ACT}$ in $\Gamma_{ACT}$ in an
asynchronous setting 
such that, for all inputs, $\vec{\sigma}_{ACT}$ and $\vec{\sigma} + \sigma_d$
produce the same set of possible distributions over outputs (note that agents have no control over how long the messages take to be delivered, and this can affect the output).
\commentout{
The case where $k=0$ can be viewed as a generalization of the BCG
result.
}

Our goal in this paper is to prove a lower bound that matches the upper
bounds of 
ADGH.  
\commentout{
BCG in fact already proved a matching lower
bound; however, their proof has a nontrivial error.%
\footnote{Ran Canetti agreed that there was a nontrivial problem with
  the BCG   proof and that a technique different from that used
  in his thesis \cite{canetti96studies} was required for the proof. We
  thank him for his comments.}
}
\commentout{
We thus
start by providing a careful proof of the lower bound for
$t$-secure computation in the asynchronous setting
(from now on we assume we are dealing with asynchronous systems unless explicitly stated otherwise).
The obvious next step is then to reduce the problem of $t$-secure
computation to that of implementing a $(k,t)$-robust mediator.  This
will give us the desired lower bound for implementation in
asynchronous systems.
}
To do so, we would like to reduce implementing $(k+t)$-secure
computation to implementing $(k,t)$-robust mediators.
If such a reduction were possible, the $n > 4(k+t)$ lower bound for implementing
$(k,t)$-robust mediators would follow immediately from the same lower
bound for secure computation \cite{ADS20,BKR94,canetti96studies}. Unfortunately,
there does not seem to be such a reduction.
However, we show that there exists a nontrivial
reduction from a slightly weaker notion of $(k+t)$-secure computation,
which we call \emph{$(k+t)$-strict secure computation}, to implementing
$(k,t)$-robust mediators. We thus start by providing a careful proof
of the lower bound for 
$(k+t)$-strict secure computation in the asynchronous setting. 
In the process, we also give a simple alternative proof for the 
lower bound on asynchronous secure computation.%
\footnote{As Ran 
Canetti [private communication] agreed,  there is a nontrivial
problem with the proof given in his thesis \cite{canetti96studies};
a different technique is needed. We 
  thank him for his comments.} 

\commentout{
On the face of
it, being able to implement 
an arbitrary $(k,t)$-robust mediator
protocol
seems to be at 
least as hard a problem as
 that of $t$-secure computation,
  so the $n > 4t$ lower bound should lead to an $n > 4(k+t)$ lower bound
 for implementing a mediator.
  Indeed, as we mentioned, $t$-secure computation is often presented
  as a problem 
of implementing a mediator. 
Unfortunately, 
there does not seem to be a trivial reduction from secure computation to
implementing $(k,t)$-robust protocols with a mediator.  However, we 
show that there exists such a reduction if we use a weaker notion of secure
computation:  
A protocol
}

Intuitively, a protocol
$t$-strictly securely computes a
function $f$ if it satisfies the properties of secure computation but
only for adversaries 
 consisting of exactly $t$ malicious agents. It
might seem that $t$-secure 
computation should be equivalent to $t$-strict secure computation.
After all, if a function can be securely computed 
with
 adversaries  of maximal size, surely it can be securely computed
with smaller
 adversaries!  As we show by example in
 Section~\ref{sec:other-secure}, this is not the case. 
Intuitively, the problem is that an adversary consisting of fewer than
$t$ agents may not be permitted to learn as much as an adversary
consisting of $t$ agents.
While investigating these issues, we noted an
ambiguity in the definition of $t$-secure computation in BCG, which
led us to consider yet another notion that we call $t$-weak secure
computation.  As the name suggests, it is weaker than $t$-secure
computation; we show that
it is 
actually
equivalent to $t$-strict secure
computation.  By considering these variants of secure 
computation,
 we
gain a deeper understanding of 
its subtleties.

\section{Basic Definitions}\label{sec:definitions}




\subsection{The Asynchronous Model}

The model used throughout this paper is the one used by
ADGH~\cite{ADGH19}, which consists of an asynchronous network in which
there is a reliable, authenticated and asynchronous channel between
all pairs of players. This means that all messages sent by player $i$
to player $j$ are guaranteed to be delivered eventually, and that $j$
can identify that these messages were sent by $i$. However, that these
messages may be delayed arbitrarily. The order in which these messages
are received is decided by an adversarial entity called the
\emph{scheduler}. The scheduler also decides in which order the
players are scheduled. 

We define the local history $h_i$ of player $i$ to be the ordered sequence of local computations (including random coin tosses), messages sent and received (including senders and recipients), in addition to all the times in between in which $i$ has been scheduled. Similarly, we define the local history $h_T$ of a subset $T$ of players as the collection of local histories $h_i$ with $i \in T$. Note that in the distributed computing literature, it is generally assumed that the players are scheduled automatically right after receiving a message. However, in this model we allow the scheduler to decide separately when the messages are delivered and when the players are scheduled. This means that whenever it is the turn of a player to act, that player may have received no messages since its last turn, or it may have received more than one (as opposed to exactly one). It is straightforward to check that all of our results also hold if we use the more standard model.

\subsection{Secure Computation}\label{sec:secure-computation}

For the main definitions in this section, we need the following
notation, largely taken from BCG.  Given a finite domain $D$, let
$\vec{x}$ be a vector in $D^n$. Given a set $C \subseteq 
[n]$, 
denote by $\vec{x}_C$ the vector obtained by projecting $\vec{x}$ onto
the indices of $C$. Also, given a vector $\vec{z} \in D^{|C|}$,
let 
$\vec{x}/_{(C, \vec{z})}$
 be the vector obtained by replacing
the entries of $\vec{x}$ indexed by 
$C$
 by the corresponding entries
of $\vec{z}$. To 
simplify notation, given a function $f: D^n \rightarrow D$, we
write $f_C(\vec{x})$ rather than 
$f(\vec{x}/_{(\overline{C},\vec{x}_0)})$
 to denote
the output of evaluating $f$ on $\vec{x}$ 
with the entries in $\vec{x}$ not indexed by an element of $C$ 
replaced 
by
some default value $x_0 \in D$.

Suppose that a group of $n$ agents wants to compute the output of a function
$f:D^n \rightarrow D$, but the $i$th input $x_i$ is known only by
agent $i$. A protocol securely computes $f$ if (a)
all 
agents
 correctly compute $f$, regardless of the
deviations of malicious players, and (b) malicious 
agents
 do not
learn anything about the input of honest 
agents
 beyond what can be
deduced from the output of $f$.
Before going on, we need to make precise what it means to 
\emph{correctly compute} $f$, since a
malicious 
agent
 can lie about its input or not participate
in the computation at all. Roughly speaking, the idea is 
to accept as correct any output of $f$ that can be obtained from an 
input profile that differs from the actual input profile
in at most $t$ coordinates
(intuitively, these coordinates are ones corresponding to inputs of malicious
agents
 who did not submit a value or lied about their actual input.)
More precisely, 
we have the following definition:
\begin{definition}
A
 protocol
$\vec{\pi}$ $t$-securely computes $f$ in synchronous systems if for
every coalition $T$ of at most $t$ malicious 
agents
and every strategy $\vec{\tau}_T$ for players in $T$, 
 there exist
functions $h: D^{|T|} \rightarrow D^{|T|}$ and $O: D^{|T|} \times D
\times T \rightarrow \{0,1\}^*$ such that,
for each input $\vec{x}$, 
\begin{itemize}
\item [(a)] each agent $i \not \in
  T$ outputs $f(\vec{x}/_{(T, h(\vec{x}_T)})$;
\item[(b)] each agent $i \in
T$ outputs $O(\vec{x}_T,f(\vec{x}/_{(T, h(\vec{x}_T)}), i)$. 
\end{itemize}
\end{definition}
Note that
$h$ and $O$ encode how malicious 
agents
 might lie about their inputs 
(if a malicious 
agent
does not participate in the computation, its input is assumed to be
the default value $x_0 \in D$)
and what they output, respectively. 
 We thus consider an output to
be correct if only the inputs of 
agents
 in $T$
used in the computation of $f$
 differ from their actual
inputs, and if the output of malicious  
agents
 output is just a function
of the output of $f$ and their own inputs. 
Note that this last requirement captures the fact that malicious
agents
 do not learn anything besides the (honest 
agents' ) output of
the secure computation protocol, since otherwise they could use this
extra information to  
generate outputs that cannot be written as such a function $O$.
Since malicious 
agents
 can
randomize, we assume that both $h$ and $O$ have an extra input $r$,
a bitstring chosen uniformly at random from $\{0,1\}^\omega$
(the set of all finite bitstrings), and
that 
agent
 $i$'s output is distributed identically to
$f(\vec{x}/_{(T, h(\vec{x}_T)})$ or  
$O(\vec{x}_T,f(\vec{x}/_{(T, h(\vec{x}_T)}), i)$, depending on whether $i$ is
honest. (See Definition~\ref{def:secure} for the more standard
formalization of this property.) 
BGW proved the following result: 

\begin{theorem}\emph{\cite{BGW88}}
  If $D$ is a finite domain, $n > 3t$, and $f:D^n
\rightarrow D$, then there exists a protocol $\vec{\pi}$
that $t$-securely computes $f$ in synchronous systems. 
\end{theorem}


Subtleties introduced by asynchrony make the definition of secure
computation slightly more involved in asynchronous systems.
In asynchronous systems, as is standard, we assume that there is a
scheduler
with its own strategy $\sigma_e$
that decides the order in which agents act and how long it
takes for a message to be delivered.
As pointed
out by ADGH, malicious 
agents
 can
effectively communicate with the scheduler, so we can assume
that the adversary and malicious 
agents
 are all controlled by a
single entity. We call this entity the \emph{adversary}; we define
it as a tuple $A = (T, \vec{\tau}_T, \sigma_e)$ consisting of the set
$T$ of malicious 
agents, 
their strategy $\vec{\tau}_T$, and the
scheduler's strategy $\sigma_e$. 
\commentout{
This implies that malicious agents
can deviate in a new way respect to synchronous systems: they can
decide (if they collude with the scheduler): when securely computing a
function $f$, agents cannot distinguish between an agent that did not
send its input from an agent that is being slow, the only guarantee is
that at least $n-t$ 
agents
 eventually send theirs. Thus, the
adversary can decide which subset $C \subseteq [n]$ of size at least
$n-t$ is taken into consideration. 
}
The existence of such adversaries implies that there are
deviations that are possible in asynchronous systems that are not
possible in synchronous systems; specifically, the
scheduler can delay a subset of 
agents
 until the other 
agents 
terminate the protocol. If the number of 
agents
 delayed is less than the
number of malicious 
agents
 that the protocol tolerates, delayed
honest 
agents
 are indistinguishable from malicious 
agents 
  that never
engage in the communication, and thus the remaining 
agents
 must be
able to terminate regardless of the delay. Since the inputs of delayed
honest 
agents
 are not taken into consideration, the
adversary can choose a set $C \subseteq [n]$ of size at least $n-t$
and force the computation to ignore the inputs of 
agents
 not in
  $C$.

To define asynchronous secure computation, BCG introduced another
type of adversary that they called a \emph{trusted-party adversary}.
A $t$-trusted-party adversary is
defined as a quadruple $A = (T, h, c, O)$ where 
\begin{itemize}
  \item $T$ is the set of malicious 
agents;  
\item $h: D^{|T|} \times \{0,1\}^\omega \rightarrow D^{|T|}$ is the input
substitution function; 
\item $c: D^{|T|} \times \{0,1\}^\omega \rightarrow \{C \subseteq [n]
  \mid |C| \ge n-t\}$ is a subset of 
agents  
   (intuitively, the ones
    whose inputs are taken into consideration);
\item 
    $O: D^{|T|} \times \{0,1\}^\omega \times D \times T \rightarrow \{0,1\}^*$
    is the output function for the malicious 
agents.    
\end{itemize}
In the sequel, we use ``trusted-party adversary'' to refer to such a
tuple $(T,h,c,O)$, and reserve the term adversary for a tuple of the form
$(A,\vec{T}_T, \sigma_e)$, as  defined earlier.

\commentout{
Let $\vec{x}$ be a vector in $D^n$. Given a set $C \subseteq [n]$,
denote by $\vec{x}_C$ the vector obtained by projecting $\vec{x}$ onto
the indices of $C$. Also, given a $|C|$-vector $\vec{z} \in D^{|C|}$,
let $\vec{x}_{(B, \vec{z})}$ be the vector obtained by replacing
the entries of $\vec{x}$ indexed by $B$ by those of $\vec{z}$. To
simplify notation, given a function $f: D^n \rightarrow D$, we
write $f_C(\vec{x})$ rather than $f(\vec{x}_{(\overline{C},
    (\vec{x}_0)_{|C|})})$ to denote
the output of evaluating $f$ on $\vec{x}$ 
with the entries in $\vec{x}$ not in $C$ 
replaced  
 by 0. 
 }

Given a function $f: D^n \rightarrow D$, a trusted-party adversary $A
= (T, h, c, O)$, and an input vector $\vec{x}$, let $C = c(\vec{x}_T,
r)$ and $\vec{y} = \vec{x}/_{(T, h(\vec{x}_T, r))}$.  Intuitively, $C$
is the set of 
agents
 whose inputs are considered and $\vec{y}$ is the input
profile obtained by replacing the actual inputs of 
agents
 in $T$
with the output of $h$. The output of $f$ with
trusted-party adversary $A$ and input $\vec{x}$ is an $n$-vector of
random variables $\vec{\rho}(A, \vec{x}; f)$ such that $$\rho_i(A,
\vec{x}; f) =  
\left\{
\begin{array}{ll}
    (C, f_C(\vec{y})) & \mbox{if } i \not \in T\\          
O(\vec{x}_B, r, f_C(\vec{y}), i) & \mbox{if } i \in T.    
\end{array}\right.$$

Note that the outputs of trusted-party adversaries are analogous to
the outputs of secure computation in the synchronous case, except that
here we must take into account the subset $C$ of 
agents
 that provide
their inputs. In asynchronous systems, secure computation is defined
as follows: 

\begin{definition}[Secure computation]\label{def:secure}
  Let $f: D^n \rightarrow D$ be a function of $n$ variables over some
 finite domain 
   $D$. The protocol $\vec{\sigma}$ \emph{$t$-securely computes
$f$ in an asynchronous setting} if the following hold for all
(standard) adversaries $A = (T, \vec{\tau}_T, \sigma_e)$ with $|T| \le t$: 
\begin{itemize}
\item on all inputs, 
agents
 not in $T$ terminate the protocol with
  probability 1;
\item there exists a 
$t$-trusted-party adversary 
$A^{tr} = (T, h, c, O)$ 
    such that, for all inputs $\vec{x} \in D^n$, we have 
  $\vec{\sigma}(\vec{x}, A) \sim \vec{\rho}(A^{tr}, \vec{x};
  f)$ (i.e.,  $\vec{\sigma}(\vec{x}, A)$ and 
  $\vec{\rho}(A^{tr}, \vec{x})$
  are identically distributed).
\end{itemize}
\end{definition}

In other words, a protocol $\vec{\sigma}$ $t$-securely computes some
function $f$ if it terminates with probability 1 and there exists a
trusted-party adversary such that that, for all inputs, gives the 
same distribution over outputs. BCG proved the following
result. 

\begin{theorem}\emph{\cite{BCG93}}\label{thm:BCG}
If $D$ is a finite domain, $n > 4t$, and $f:D^n
\rightarrow D$, then there exists a protocol $\vec{\pi}$
that $t$-securely computes $f$ in asynchronous systems.
\end{theorem}

\subsection{Weaker Notions of Secure Computation}\label{sec:other-secure}

Note that the $T$ in the second condition of
Definition~\ref{def:secure}, that is, the $T$ in the trusted-party
adversary 
$(T,h,c,O)$,
 is the same as the $T$ in the adversary.
This is also true in the BGW definition of $t$-secure computation.
While we believe that this was also the intention of BCG, 
their definition
simply says that that there 
exists
a 
$t$-trusted-party adversary,
 without specifying $T$ (the set of
  malicious agents)
that satisfies the second bullet of Definition~\ref{def:secure}.
   Taking this definition seriously leads to a
slightly weaker notion of secure 
computation
 that we call 
  $t$-\emph{weak secure computation}, which is defined just as
$t$-secure computation
 except that the 
$t$-trusted-party adversary 
$A^{tr}$ may involve any subset $T'$ of malicious agents such that 
 $|T'| = t$ and $T' \supseteq T$, as opposed to consisting of the same
set $T$ of malicious agents as $A$.  

We show next that $t$-weak secure computation is strictly weaker than
the standard notion of secure computation. To do so, first we
introduce an intermediate notion of secure computation called
\emph{$t$-strict secure computation}; it is defined just as $t$-secure
computation, except that we require only that the
properties are satisfied for adversaries of size exactly $t$ (i.e., for
$|T| = t$).
\commentout{
As we mentioned in the introduction, somewhat surprisingly, $t$-strict
secure computation is strictly weaker than $t$-secure computation, but
is equivalent to $t$-weak secure computation.
}
As we mentioned in the introduction, somewhat surprisingly, $t$-strict
secure computation is strictly weaker than $t$-secure computation,
but, as we show next, it is actually equivalent to $t$-weak secure
computation.

\begin{theorem}\label{thm:ranking}
\leavevmode
\begin{itemize}
\item[(a)] If a protocol $\vec{\pi}$ $t$-securely computes a function $f$, it also $t$-strictly securely computes $f$.
\item[(b)] A protocol $\vec{\pi}$ $t$-strictly securely computes a
function $f$ iff it $t$-weakly securely computes $f$. 
\item[(c)] 
If $t > 1$ and $n > 4t$, there
 exists a function $f$ on $n$ variables 
  and a protocol
$\vec{\pi}$ such that $\vec{\pi}$ $t$-strictly securely computes $f$
but does not $t$-securely computes $f$.  
\end{itemize}
\end{theorem}

\begin{proof}
  Part (a) follows immediately from the definition of secure computation
  and strict secure computation. For part (b), first suppose that a
protocol $\vec{\pi}$ $t$-stricly securely computes $f:D^n
\rightarrow D^n$.   Given an adversary $A = (T,
\vec{\tau}_T, \sigma_e)$ with $|T| \le t$, consider an adversary of the
form $A' = (T \cup T', \vec{\tau}_T + \vec{\pi}_{T'}, \sigma_e)$
(Since the agents in $T'$ play $\vec{\pi}$, they in fact do not deviate.)
Because $\vec{\pi}$
such that $T \cap T' = \emptyset$ and $|T \cup T'| = t$. Since $\vec{\pi}$
$t$-strictly securely computes $f$, there exists a trusted-party
adversary $A^{tr} = (T \cup T', h, c, O)$ such that
$\vec{\pi}(\vec{x}, A') = \vec{\rho}(A^{tr}, \vec{x};f)$. By
construction, 
$\vec{\pi}(\vec{x}, A') \sim \vec{\pi}(\vec{x}, A)$,
 since
the additional malicious agents in $A'$ do not deviate from the
protocol. Therefore,  
$\vec{\pi}(\vec{x}, A) \sim \vec{\rho}(A^{tr},
\vec{x};f)$, so  $\vec{\pi}$ $t$-weakly securely computes $f$. 

The converse is almost immediate from the definitions.  Suppose that
protocol $\vec{\pi}$ $t$-weakly securely computes $f:D^n 
\rightarrow D^n$ for some $t$. Given an adversary $A = (T,
\vec{\tau}_T, \sigma_e)$ with $|T| = t$, then, by assumption, there is
a $t$-trusted-party adversary $A' = (T,h, c, O)$ such that
$\vec{\pi}(\vec{x}, A') = \vec{\pi}(\vec{x}, A)$.

For part (c), consider the following setup. Let $\mathbb{F}_2$ be
the field with domain $\{0,1\}$.   
Given $n$ and $t$ such that $t > 1$ and $n > 4t$, consider a function
$f:  (\mathbb{F}_2)^{n^3} \rightarrow (\mathbb{F}_2)^{n^2}$ that does
the following.
Given the input 
$(x^i, c^i, y^i, z^i) \in (\mathbb{F}_2)
^{n^2}$
 of each 
 agent
$i$, where $x^i\in \mathbb{F}_2$, $c^i \in (\mathbb{F}_2)^n$, $y^i \in
(\mathbb{F}_2)^{t-1}$, and $z^i$ consists of the remaining 
$n^2 - n - t$
 coordinates
(which 
do not affect
the function 
$f$; they are needed 
because in Definition~\ref{def:secure}, the input space of each 
agent
must be the same as the output space of the function), 
 let $p_i \in \mathbb{F}_2[X]$ be the
unique polynomial of degree $t-1$ such that $p_i(0) = x^i$ and $p_i(j)
= y^i_j$ for all $j = 1,2, \ldots, t-1$. The output of 
$f$ is then $\{p_i(j) + c^j_i\}_{i,j \in [n]}$. In other words, $f$
encodes the first coordinate of each 
agent's
 input using Shamir's 
agent
secret sharing scheme \cite{shamir}. The polynomial $p_i$ that each 
agent
$i$ uses to do the encoding 
and the one-time pads $c^j_i$ added by $i$ to each of the shares
are part of $i$'s input, and not known by the
other 
agents.
 However, a coalition $T$ of $t$ malicious 
agents 
  can
reconstruct the values $p_i(j)$ for all $i \in [n]$ and $j \in T$, and
thus is able to reconstruct each $x_i$ as well, since the 
agents
in $T$ know
$t$ points on each polynomial $p_i$, although no coalition of size
strictly smaller than $t$ knows those values.

Consider a protocol $\vec{\pi}$ that consists of the following: each
agent 
$i$ performs its part of BCG's $t$-secure computation protocol
to compute $f$ and then, if $i$ is included in the core set of the
output, $i$ broadcasts the first bit of its input. By the
earlier argument, if the adversary is of size exactly $t$, it can
reconstruct the first coordinate of the inputs of the 
agents
 in the
core-set from the output of $f$ and its own inputs, which means that
the values broadcast after BCG's secure computation protocol do not
give any extra information about the inputs of honest 
agents
 to the
adversary. However, this is not true for smaller adversaries. 
Thus, $\vec{\pi}$ 
 $t$-strictly securely computes $f$, but does not
$t$-securely compute $f$. 
\end{proof}

\subsection{Implementing mediators}

We now formalize the notion of $(k,t)$-robust equilibrium.  Recall
that in this setting, there are three games, an underlying game
$\Gamma$ for $n$ players, which is technically a \emph{Bayesian game},
a mediator game $\Gamma_d$, and a communication game $\Gamma_{ACT}$.
In a Bayesian 
game, players have inputs, and their payoff depends on the profile of
moves made and the input profile.
The set of players is the same in all three games, except that in the
mediator game, there is also a mediator, who can be viewed as a special
non-strategic player (i.e., there is no utility function for the
mediator) and uses a commonly-known strategy, denoted $\sigma_d$.
In the mediator game, the players just communicate with the mediator
(although deviating or malicious players are allowed to communicate
with each other). 
 In the communication game, they communicate among
themselves using a point-to-point network.  After communicating in the
mediator game and the communication game, the
players make a move in the underlying game $\Gamma$, and get payoffs
as in $\Gamma$.
As is standard, we use $\vec{\sigma} := (\sigma_1, \ldots, \sigma_n)$ to
denote a strategy profile for $n$ players in which each player $i$
plays $\sigma_i$; we use $\vec{\sigma} + \sigma_d$ to denote the
strategy profile for $n$ players and a mediator in which each player
$i$ plays $\sigma_i$ and the mediator plays $\sigma_d$;
finally, we use $(\sigma_{-T},\tau_{T})$ to denote the strategy where
each player $i \notin T$ uses the strategy $\sigma_i$ while $j \in T$
uses the strategy $\tau_j$.

In this game-theoretic setting, we are interested in protocols that
are \emph{$(k,t)$-robust}.  To define $(k,t)$-robustness, we need two
preliminary definitions.

\begin{definition}
    Given a game $\Gamma$, a strategy profile $\vec{\sigma}$ is
    \emph{$t$-immune} if for 
all subsets $T$ of size at most $t$ and all strategies $\vec{\tau}_T$
for players in $T$ $u_i(\vec{\sigma}_{-T}, \vec{\tau}_T) \ge
u_i(\vec{\sigma})$ for all $i \not \in T$, where $u_i(\vec{\sigma})$
is the payoff of player $i$ when players play $\vec{\sigma}$.
\end{definition}

Intuitively, a strategy profile is a $t$-immune equilibrium if no
subset of at most $t$  players can decrease the payoff of other
players by deviating, 

\begin{definition}
A strategy profile $\vec{\sigma}$ is a
\emph{$(k,t)$-resilient} (resp., \emph{strongly $(k,t)$-resilient})
\emph{equilibrium} of a game $\Gamma$
if, for all disjoint subsets $K$
and $T$ of sizes at most $k$ and $t$, respectively, and all strategy profiles
$\vec{\tau}_{K \cup T}$ for players in $K \cup T$,
$u_i(\vec{\sigma}_{-(K \cup T)}, \vec{\tau}_{K \cup T}) \le
u_i(\vec{\sigma}_{-T}, \vec{\tau}_{T})$ for some (resp., for all) $i
\in K$. 
\end{definition}

Intuitively, a strategy protocol is a $(k,t)$-resilient if no subset of
at most  $k$ players can all increase their payoffs, even if 
they can collude with up to $t$ malicious players.  It is a strong
$(k,t)$-resilient equilibrium if not even one player in the set can
increase its payoff.

\begin{definition}
 A strategy profile is a \emph{$(k,t)$-robust} (resp., \emph{strongly
$(k,t)$-robust}) \emph{equilibrium} in a game $\Gamma$ if it is $t$-immune and
a $(k,t)$-resilient (resp., 
strongly $(k,t)$-resilient) equilibrium. 
\end{definition}

The notion of
$(k,t)$-robustness was introduced by Abraham, Dolev, Gonen and
Halpern \citeyear{ADGH06}, who also proved the following: 

\begin{theorem}\emph{\cite{ADGH06}}
  If $\vec{\sigma} + \sigma_d$ is a $(k,t)$-robust equilibrium for a
  synchronous game $\Gamma_d$ that extends some game $\Gamma$ and $n >
  3(k+t)$, then there exists a $(k,t)$-robust equilibrium $\vec{\sigma}_{ACT}$
for $\Gamma_{ACT}$ such that, for all input profiles,  the
distribution over outcomes induced by $\vec{\sigma} + \sigma_d$ is identical
to that induced by $\vec{\sigma}_{ACT}$. 
\end{theorem}

ADGH proved an 
analogous result for asynchronous systems. Making the statement
precise required a little care since, even for a fixed input, the output
distribution induced by a protocol depends on the scheduler.  This observation
motivates the following definition.

\begin{definition}
  Protocol $\vec{\sigma}$  \emph{implements} protocol $\vec{\tau}$ in an
asynchronous network if, for all input profiles $\vec{x}$ and all
schedulers $\sigma_e$, there exists a scheduler $\sigma_e'$ such that
the distribution over output profiles induced by $\vec{\sigma}$ with
input profile
$\vec{x}$ and scheduler $\sigma_e$ is identical to the
distribution over output profiles induced by $\vec{\tau}$ with
input profile $\vec{x}$ and scheduler $\sigma_e'$.
\end{definition}
Essentially, this definition says that $\vec{\sigma}$ implements
$\vec{\tau}$ if, for all input profiles $\vec{x}$, the set of possible
output distributions of $\vec{\sigma}$ is the same as that of $\vec{\tau}$.


\begin{theorem}\emph{\cite{ADGH19}}\label{thm:ADGH19}
If $\vec{\sigma} + \sigma_d$ is a $(k,t)$-robust strategy for an
game $\Gamma_d$ that extends some game $\Gamma$ and $n >
4(k+t)$, then there exists a $(k,t)$-robust protocol $\vec{\sigma}_{ACT}$
for $\Gamma_{ACT}$ that implements $\vec{\sigma} + \sigma_d$. 
\end{theorem}

It is easy to 
$t$-securely compute a function $f$ with the help of a
mediator: Each player sends its input to the mediator, the mediator
waits until it receives an input from at least $n-t$ 
agents (in synchronous systems it just waits one round), then it computes the
output of $f$ given the input of the players, and sends it to all
players.
\commentout{
It follows that a lower bound for $t$-secure computation gives a lower
bound for $($
}
However, despite the fact that we think of $t$-secure computation in
terms of mediators, 
it is not obvious that Theorem~\ref{thm:BCG}
follows from Theorem~\ref{thm:ADGH19}, due to the differences
between the definitions of $(k,t)$-robustness and secure
computation. 
\commentout{
We show over the next sections how to reduce $t$-secure computation to
implementing $(t,0)$-robust strategies in mediator games, and how to
reduce $(t+k)$-weak secure computation to implementing $(k,t)$-robust
strategies in mediator games. We use both of these reductions to prove
matching lower bounds for Theorem~\ref{thm:BCG} and
Theorem~\ref{thm:ADGH19}. 
}
At the end of Section~\ref{sec:reduction}, we sketch how to reduce
$t$-secure computation to implementing $(t,0)$-robust strategies in
mediator games. 

\section{Main Results}

In this paper we show that the bound in Theorem~\ref{thm:ADGH19} is tight:

\begin{theorem}\label{thm:main2}
If $k+t+1 < n \le 4k+4t$ there exists a $(k,t)$-robust 
(resp., strongly $(k,t)$-robust) strategy profile $\vec{\sigma} + \sigma_d$ 
for $n$ players and a mediator such that there is no $(k,t)$-robust
(resp., strongly $(k,t)$-robust) strategy profile
$\vec{\sigma_{ACT}}$ that implements $\vec{\sigma} + \sigma_d$. 
\end{theorem}

The proof of Theorem~\ref{thm:main2} is divided in two parts. 

\subsection{Case 1: $3k + 3t \le n \le 4k + 4t$}

Here, we show that that $(k+t)$-strictly securely computing a function $f$ reduces to implementing a $(k,t)$-robust strategy $\vec{\sigma} + \sigma_d$ for some game $\Gamma^{f,k,t}$. To make this precise, we need the following definition:

\begin{definition}
If $g: A \rightarrow B$, and  $\sigma$ is a strategy that plays
actions in $A$, then $g(\sigma)$ is the strategy that is
identical to $\sigma$ except that each action $a \in A$ is replaced by
$g(a) \in B$.  If $\vec{\sigma}$ is a strategy profile where each
player $i$ plays actions in $A$, then $g(\vec{\sigma}) = 
(g(\sigma_1), \ldots, g(\sigma_n))$.  
\end{definition}

\begin{theorem}\label{thm:secure-computation-reduction}
    If $f:D^n \rightarrow D$, $D$ is a finite domain,
and $2(k+t) < n$, then there
exists a game $\Gamma_d^{f,k,t}$ in which all players have the same set $A$
of possible actions, a function $g: A \rightarrow D$, and a
$(k,t)$-robust strategy 
(resp., strongly $(k,t)$-robust strategy)
$\vec{\sigma} + \sigma_d$ for $n$ players and
the mediator in $\Gamma_d^{f,k,t}$ such that if $\vec{\sigma}_{ACT}$ is
 $(k,t)$-robust implementation
(resp., strongly $(k,t)$-robust implementation) 
  of $\vec{\sigma} +\sigma_d$,
 then $g(\vec{\sigma}_{ACT})$ $(k+t)$-strictly securely
computes $f$. 
\end{theorem}

Theorem~\ref{thm:secure-computation-reduction} shows that being able to implement all $(k,t)$-robust mediators with $n$ players implies that all functions on $n$ variables can be $(k+t)$-strictly securely computed. The proof of Theorem~\ref{thm:main2} for $3(k+t) \le n \le 4(k+t)$ follows from the fact that if $3(k+t) \le n \le 4(k+t)$ there exist functions that cannot be $(k+t)$-weakly securely computed:

\begin{theorem}\label{prop:weak-secure-computation}
\leavevmode
\begin{itemize}
\item[(a)] If $n > 4t$ or $n \le 2t$, every function $f:D^n \rightarrow D$
    can be $t$-weakly securely computed
  in asynchronous systems. 
      \item[(b)] If 
$3t \le n \le 4t$, 
   there exists a domain $D$ and a function $f: D^n \rightarrow D$ that cannot be $t$-weakly securely computed
  in asynchronous systems.  
   \end{itemize}
\end{theorem}

The proof of Theorem~\ref{thm:weak-secure-computation} is given in Section~\ref{sec:lower-secure-comp}. A slight variation of the proof provides a simple proof for the well-known lower bound for secure computation on asynchronous systems:

\begin{theorem}\label{prop:secure-computation}
\leavevmode
\begin{itemize}
\item[(a)] If $n > 4t$ or $n \le t$, for all domains $D$, every
  function $f:D^n \rightarrow D$ can be $t$-securely computed
  in asynchronous systems. 
\item[(b)] If $t < n \le 4t$ there exists a domain $D$ and a function $f:D^n \rightarrow D$ that cannot be $t$-securely computed
  in asynchronous systems.  
\end{itemize}
\end{theorem}

\subsection{Case 2: $k+t+1 < n \le 3k + 3t$}

If $k + t + 1< n \le 3k + 3t$ we show that implementing $(k+t)$-resilient weak consensus with $n$ players can be reduced to implementing $(k,t)$-robust mediators:

\begin{theorem}\label{thm:reduction-to-weak2}
If $n > k+t+1$, then there exists a game
$\Gamma^{k,t}_d$ in which all players have the same set $A$ of
possible actions, a function $g: A \rightarrow \{0,1\}$, and a
$(k,t)$-robust (resp., strongly $(k,t)$-robust) strategy $\vec{\sigma} + \sigma_d$ for $n$ players and
the mediator in $\Gamma_d^{k,t}$ such that if $\vec{\sigma}_{ACT}$ is
a $(k,t)$-robust (resp., strongly $(k,t)$-robust) implementation of $\vec{\sigma} + \sigma_d$,
then $g(\vec{\sigma}_{ACT})$ is a $(k+t)$-resilient implementation of weak consensus. 
\end{theorem}

The proof of Theorem~\ref{thm:main2} for $k+t+1 \le n \le 3(k+t)$ follows from Lamport's lower bound for weak consensus~\citeyear{cite:L83}.

\section{Proof of Theorems~\ref{prop:weak-secure-computation} and \ref{prop:secure-computation}}\label{sec:lower-secure-comp}

In this section we prove the
$n > 4t$
 lower bound for $t$-secure
 computation
(also proved or claimed in \cite{ADS20,BKR94,canetti96studies}),
 which we strengthen slightly by showing
that it applies to $t$-weak secure computation.

Our proof is similar to 
that of Canetti \citeyear{canetti96studies}
 at a high level: We construct a function $f$
with four inputs, the scheduler schedules the
agents 
 so that the
fourth 
agent
 never gets to participate in the computation, and one of
the three remaining 
agents 
 is malicious and manages to trick the
other two participating 
agents 
 into outputting something inappropriate.
  Canetti then claims that
\emph{conversations} between 
agents 
 (where a conversation is just the
collection of 
messages sent by two given 
agents)
 must be independent of the inputs
of the 
agents, 
agent
 3 can send messages to 
agents  
  1 and 2 in such
a way that 
agents 
 1 and 2 believe they should output different
values. However, there are two significant problems with this approach: 
\begin{itemize}
\item[(a)] First, the conversations between the 
agents 
 might not be
    totally independent of their inputs, since they can depend on the
  output of the computation, and this ultimately does depend on the
  inputs. For example, 
agents   
can run Bracha's \citeyear{Br84}  consensus protocol 
(which tolerates $t$ malicious 
agents   
   if $n > 3t$) after
  terminating the secure computation protocol to decide the
  output. This would guarantee that all honest 
agents   
   output the same
      value at the end of the computation, so their conversations are
      certainly not independent. 
\item[(b)] Second, there is a more subtle issue when trying to
  simultaneously trick 
agents   
   1 and 2 into outputting some given
  values $a$ and $b$, respectively. Even though 
Canetti proves 
that for the function $f$ that he uses and a particular input
$\vec{x}$, for each
  conversation $h_{1,2}$ between 1 and 2, there is a protocol for
  player 3 that results in a conversation
  $h_{1,3}$ between 1 and 3 such that 1 outputs $a$, and that for each
    conversation $h_{1,2}$ between 1 and 2 there exists a protocol for
    player 3 that results in  a conversation
  $h_{2,3}$ between 2 and 3 such that 2 outputs $b$, there might not
        exist a protocol for agent 3 that results in 1 and 2 having
        conversation $h_{1,3}$ and agents 2 and 3 having conversation 
  $h_{2,3}$  simultaneously. In fact, if $a$ and $b$ are
    different and 
agents     
run a consensus protocol as in (a), there is not. 
\end{itemize}

Roughly speaking, we deal with these issues as follows. We prove that
for our function $f$, a malicious 
agent
 can make all honest 
agents 
output the same incorrect value, 
and we show that in our case there does exist a conversation $h_{1,2}$
between 1 and 2 such that 
agent
 3 can trick both of them
simultaneously, as desired (see Lemma~\ref{lemma:same-histories}). 
Some of these techniques can also be applied to prove lower bounds for
weak secure computation.

\section{END}

\commentout{
We start by carefully proving the 
}
In this section we prove the
$n > 4t$
 lower bound for $t$-secure
 computation
(also proved or claimed in \cite{ADS20,BKR94,canetti96studies}),
 which we strengthen slightly by showing
that it applies to $t$-weak secure computation.

Our proof is similar to 
that of Canetti \citeyear{canetti96studies}
 at a high level: We construct a function $f$
with four inputs, the scheduler schedules the
agents 
 so that the
fourth 
agent
 never gets to participate in the computation, and one of
the three remaining 
agents 
 is malicious and manages to trick the
other two participating 
agents 
 into outputting something inappropriate.
  Canetti then claims that
\emph{conversations} between 
agents 
 (where a conversation is just the
collection of 
messages sent by two given 
agents)
 must be independent of the inputs
of the 
agents, 
agent
 3 can send messages to 
agents  
  1 and 2 in such
a way that 
agents 
 1 and 2 believe they should output different
values. However, there are two significant problems with this approach: 
\begin{itemize}
\item[(a)] First, the conversations between the 
agents 
 might not be
    totally independent of their inputs, since they can depend on the
  output of the computation, and this ultimately does depend on the
  inputs. For example, 
agents   
can run Bracha's \citeyear{Br84}  consensus protocol 
(which tolerates $t$ malicious 
agents   
   if $n > 3t$) after
  terminating the secure computation protocol to decide the
  output. This would guarantee that all honest 
agents   
   output the same
      value at the end of the computation, so their conversations are
      certainly not independent. 
\item[(b)] Second, there is a more subtle issue when trying to
  simultaneously trick 
agents   
   1 and 2 into outputting some given
  values $a$ and $b$, respectively. Even though 
Canetti proves 
that for the function $f$ that he uses and a particular input
$\vec{x}$, for each
  conversation $h_{1,2}$ between 1 and 2, there is a protocol for
  player 3 that results in a conversation
  $h_{1,3}$ between 1 and 3 such that 1 outputs $a$, and that for each
    conversation $h_{1,2}$ between 1 and 2 there exists a protocol for
    player 3 that results in  a conversation
  $h_{2,3}$ between 2 and 3 such that 2 outputs $b$, there might not
        exist a protocol for agent 3 that results in 1 and 2 having
        conversation $h_{1,3}$ and agents 2 and 3 having conversation 
  $h_{2,3}$  simultaneously. In fact, if $a$ and $b$ are
    different and 
agents     
run a consensus protocol as in (a), there is not. 
\end{itemize}

Roughly speaking, we deal with these issues as follows. We prove that
for our function $f$, a malicious 
agent
 can make all honest 
agents 
output the same incorrect value, 
and we show that in our case there does exist a conversation $h_{1,2}$
between 1 and 2 such that 
agent
 3 can trick both of them
simultaneously, as desired (see Lemma~\ref{lemma:same-histories}). 
Some of these techniques can also be applied to prove lower bounds for
weak secure computation. 
\begin{theorem}\label{prop:secure-computation}
\leavevmode
\begin{itemize}
\item[(a)] If $n > 4t$ or $n \le t$, for all domains $D$, every
  function $f:D^n \rightarrow D$ can be $t$-securely computed
  in asynchronous systems. 
\item[(b)] If $t < n \le 4t$ there exists a domain $D$ and a function $f:D^n \rightarrow D$ that cannot be $t$-securely computed
  in asynchronous systems.  
\end{itemize}
\end{theorem}

\begin{theorem}\label{prop:weak-secure-computation}
\leavevmode
\begin{itemize}
\item[(a)] If $n > 4t$ or $n \le 2t$, every function $f:D^n \rightarrow D$
    can be $t$-weakly securely computed
  in asynchronous systems. 
      \item[(b)] If 
$3t \le n \le 4t$, 
   there exists a domain $D$ and a function $f: D^n \rightarrow D$ that cannot be $t$-weakly securely computed
  in asynchronous systems.  
   \end{itemize}
\end{theorem}

\commentout{
The proof of Theorem~\ref{prop:secure-computation} and
\ref{prop:weak-secure-computation} can be found in
Appendix~\ref{sec:appendix1}. 
}
To prove part (b) of
Theorems~\ref{prop:secure-computation} and
\ref{prop:weak-secure-computation} for $3t \le n \le 4t$ we show
that the majority function $f^4: \{0,1,\bot\}^4 \rightarrow \{0,1,\bot\}$
that outputs $1$ if there are at least as many inputs equal to 1 as
inputs equal to 0, and outputs 0 otherwise, cannot be $1$-weakly
securely computed. In fact, we show that  player 3 can get players 1
and 2 to output 1 even when all agents have input 0. The full proof
can be found in Appendix~\ref{sec:appendix1}. 

\section{Reducing $(k+t)$-Strict Secure Computation to Implementing a
$(k,t)$-Robust Equilibrium}\label{sec:reduction}

In this section we show that $(k+t)$-strictly securely computing a
function $f$ reduces to implementing a $(k,t)$-robust strategy
$\vec{\sigma} + \sigma_d$ for some game $\Gamma_d^{f,k,t}$. To make
this precise, we need the following definition: 

\begin{definition}
If $g: A \rightarrow B$, and  $\sigma$ is a strategy that plays
actions in $A$, then $g(\sigma)$ is the strategy that is
identical to $\sigma$ except that each action $a \in A$ is replaced by
$g(a) \in B$.  If $\vec{\sigma}$ is a strategy profile where each
player $i$ plays actions in $A$, then $g(\vec{\sigma}) = 
(g(\sigma_1), \ldots, g(\sigma_n))$.  
\end{definition}


\begin{theorem}\label{thm:secure-computation-reduction}
    If $f:D^n \rightarrow D$, $D$ is a finite domain,
and $2(k+t) < n$, then there
exists a game $\Gamma_d^{f,k,t}$ in which all players have the same set $A$
of possible actions, a function $g: A \rightarrow D$, and a
$(k,t)$-robust strategy 
(resp., strongly $(k,t)$-robust strategy)
$\vec{\sigma} + \sigma_d$ for $n$ players and
the mediator in $\Gamma_d^{f,k,t}$ such that if $\vec{\sigma}_{ACT}$ is
a
 $(k,t)$-robust implementation
(resp., strongly $(k,t)$-robust implementation) 
  of $\vec{\sigma} +\sigma_d$,
 then $g(\vec{\sigma}_{ACT})$ $(k+t)$-strictly securely
computes $f$. 
\end{theorem}

The proof of Theorem~\ref{thm:secure-computation-reduction} is
surprisingly nontrivial. Given a function $f$,
it is easy to check that there is a
$(k,t)$-robust strategy $\vec{\sigma} + \sigma_d$ with a mediator
that $(k+t)$-securely computes $f$, where we assume that actions in the
mediator game have the form $(o,tp)$, where $o$ is a possible output
of $f$ and $tp$ is the player's self-declared type (i.e., whether the player is honest, rational, or malicious), honest players
get a payoff of 1 if they all agree on 
a valid
 output of $f$ and 0
otherwise, and self-declared rational players get a payoff of 1 if
they disrupt the output of honest players:
players send their inputs
to the mediator, the mediator waits to receive $n-k-t$ inputs, sends
to every player the output of the computation (taking the remaining
$k+t$ inputs to be $\bot$), and then the players
play the output received. We would expect
that any $(k,t)$-robust implementation of $\vec{\sigma} + \sigma_d$
also $(k+t)$-securely computes $f$ (without the mediator),
but this is
not the case.  For example, if players perform a secure computation of $f$
and, right after that, they broadcast their inputs, the resulting protocol
would still be a $(k,t)$-robust implementation of $\vec{\sigma} +
\sigma_{d}$. However, using this strategy, all rational and malicious
players would learn the honest players outputs. This example shows
that the game must somehow 
reward players that declare themselves to be rational if they
manage to learn something that they shouldn't. A more detailed
discussion of this issue, our solution, and a full proof of
Theorem~\ref{thm:secure-computation-reduction} can be found in
Appendix~\ref{sec:appendix2}. 
\commentout{
The proof of Theorem~\ref{thm:secure-computation-reduction} can be
found in Appendix~\ref{sec:appendix2}. 
}
As an immediate corollary of
Theorems~\ref{prop:weak-secure-computation} and
\ref{thm:secure-computation-reduction}, we get the desired lower bound
for 
implementing mediators.

\begin{corollary}\label{thm:main}
If 
$3k + 3t \le n \le 4k + 4t$ 
 there exists a 
$(k,t)$-robust (resp., strongly robust)
 protocol $\vec{\sigma} + \sigma_d$ for $n$ players and a
mediator such that there is no
$(k,t)$-robust (resp., strongly robust)
protocol $\vec{\sigma}_{ACT}$ that implements $\vec{\sigma} +
\sigma_d$. 
\end{corollary}

\section{The Lower Bound on Implementing Mediators}

In the previous section, we showed that $(t+k)$-strict secure
computation can be reduced to implementing certain $(k,t)$-robust (or
strongly robust) strategies, and thus that if $3(t+k) \le n \le
4(t+k)$, then there exist $(k,t)$-robust (resp., strongly robust) strategies
with a mediator that cannot be implemented without a mediator. In this section,
we use a different construction to extend this impossibility result to
$t+k+1 < n \le 3k+3t$.   That is, we have the following strengthening
of Corollary~\ref{thm:main}:

\begin{theorem}\label{thm:main2}
  If $k+t+1 < n \le 4k+4t$ there exists a $(k,t)$-robust 
(resp., strongly $(k,t)$-robust) strategy profile $\vec{\sigma} + \sigma_d$ 
  for $n$ players and a mediator such that there is no $(k,t)$-robust
(resp., strongly $(k,t)$-robust) strategy profile
$\vec{\sigma_{ACT}}$ that implements $\vec{\sigma} + \sigma_d$. 
\end{theorem}

Corollary~\ref{thm:main} shows that  Theorem~\ref{thm:main2} holds if
$3k + 3t \le n \le 4k+4t$. We prove the remaining cases by reducing
weak consensus to implementing mediators, much like as we did in the
previous section for secure computation. 

\begin{theorem}\label{thm:reduction-to-weak2}
If $n > k+t+1$, then there exists a game
$\Gamma^{k,t}_d$ in which all players have the same set $A$ of
possible actions, a function $g: A \rightarrow \{0,1\}$, and a strongly
$(k,t)$-robust strategy $\vec{\sigma} + \sigma_d$ for $n$ players and
the mediator in $\Gamma_d^{k,t}$ such that if $\vec{\sigma}_{ACT}$ is
a strongly $(k,t)$-robust implementation of $\vec{\sigma} + \sigma_d$,
then $g(\vec{\sigma}_{ACT})$ is a $(k+t)$-resilient implementation of weak consensus. 
\end{theorem}

The proof of Theorem~\ref{thm:reduction-to-weak2} can be found in
Appendix~\ref{sec:appendix3}.

Since, as proved by Lamport~\citeyear{L83}, there are no $(t+k)$-resilient
implementations of weak consensus if $n \le 3(t+k)$, it follows from
Theorem~\ref{thm:reduction-to-weak2} that Theorem~\ref{thm:main2}
holds for $t+k+1 < n \le 3k+3t$ as well, completing its proof. 

\section{Conclusion}
We have shown that both $(k+t)$-secure computation and the problem of
implementing a $(k,t)$-robust equilibrium with a mediator have a
lower bound of $n > 4k + 4t$. Moreover, we have shown that this is
also a lower bound for 
weaker notions of secure computation such as $(k+t)$-strict secure
computation and $(k+t)$-weak secure computation.
Finally, by considering a number of variants of the definition of
secure computation, we also 
highlighted
some of the subtleties in
the definition.

ADGH showed that protocols can tolerate more
malicious behavior if honest players can punish rational players if
they are caught deviating. Honest players can perform this punishment
by playing an action profile that results in 
all players getting an expected payoff that is worse than their
payoff in equilibrium. Not all games have such a punishment profile, 
but
ADGH showed that for games that do,
 every $(k,t)$-robust strategy with a mediator
can be implemented if $n > 3k + 4t$.  Finding a matching lower bound for this
case remains an open problem. 


\newpage

\appendix

\section{Proof of Theorems~\ref{prop:secure-computation} and \ref{prop:weak-secure-computation}}\label{sec:appendix1}

For Theorem~\ref{prop:secure-computation}(a), 
note that if $n > 4t$, Theorem~\ref{thm:BCG} shows that every function
$f:D^n \rightarrow D$ can be $t$-securely computed, and thus $t$-weak
securely computed as well. If $n \le t$, let $\bot$ be the input
assigned to the agents that did not submit an input. It can be easily
shown that the protocol where each agent sends no messages and
outputs $(\emptyset, f(\bot^n))$ $t$-securely computes
$f$. Similarly, for Theorem~\ref{prop:weak-secure-computation}(a),
it can be easily checked that if $n \le 2t$, the protocol where
each agent sends nothing and outputs $([n-t], f(\bot^n))$
$t$-weak securely computes $f$.

For the lower bounds (Theorems~\ref{prop:secure-computation}(b) and
\ref{prop:weak-secure-computation}(b)), we proceed as follows.
Consider the function
$f^n: \{0,1, \bot \}^n \rightarrow \{0,1, \bot\}$ that essentially
takes majority between $0$ and $1$: it outputs 1 if the number of
agents with input $1$ is greater or equal to the number of agents with
input $0$, otherwise it outputs $0$. Players who do not submit an
input are assumed to have input $\bot$. 
We 
start by showing 
 that 
$f^4$ 
 cannot be 
 $1$-weakly securely computed
  by four 
agents. 

Suppose that 
$f^4$
 can be 
 $1$-weakly securely computed 
  using a protocol
$\vec{\sigma}$.
Let $\sigma_e^N$ be the scheduler that schedules 
agents 
 1, 2,
and 3 cyclically, and right before scheduling
an agent, it
delivers the messages that were sent by the other 
agents 
the last time they were scheduled. After scheduling each of the first three
agents 
 $N$ times, it schedules 
agent 
  4 as well, adding it to the
cyclic order.

Given a history $H$,
let $\vec{x}_H$ denote the input profile of 
agents 
 in $H$,
let $H_i$ denote agent $i$'s local history in $H$, 
let $H_e$ denote the scheduler's local history in $H$, and 
let $H_{(i,j)}$ denote the \emph{conversation} between 
agents 
 $i$ and
$j$ (i.e.,  
the messages sent and received between $i$ and $j$, in
addition to the
relative 
 times at which $i$ and $j$ were scheduled).
%
We can now prove essentially what BCG claimed to prove (although, as
we said, this claim does not hold for the BCG construction).

\begin{lemma}\label{lemma:same-histories}
There exist $N$ and two (finite) histories $H$ and $H'$ of $\vec{\sigma}$
where the scheduler uses $\sigma_e^N$, 
$\vec{x}_H = (1,0,1,1)$, $\vec{x}_{H'} = (0,1,1,1)$, 
agents 
 1, 2, and 3 all output 1
in $H$,  
agent
 4 is never scheduled in either $H$ or
$H'$, $H_{1,2} = H'_{1,2}$, and $H_e = H'_e$.
\end{lemma}

To prove Lemma~\ref{lemma:same-histories}, we first need to prove 
what seems to be obvious: if all agents are honest, at most $t$
agents have input 0, and $n \ge 3t$, then the output of a weakly
secure computation of $f^n$ will be 1.  While
this seems obvious (and is true), it is not quite so trivial. For example,
it is not true if 
$n = 3t-1$.
In this case, if we consider a trusted-party adversary $A = (T, c, h,
O)$, in which $|T| = t$, $h$ replaces all inputs of malicious players
with 0, and $c$ chooses all $t$ malicious players and $t-1$
additional honest players, it is easy to check that the output of
honest players is 0. 

\begin{lemma}\label{prop:technical-1}
  Let $n \ge 3t$ and let $\vec{\sigma}$ be a protocol that $t$-weakly
  securely computes $f^n$. Then for all schedulers, in all histories of
$\vec{\sigma}$ in which all agents are honest and at most $t$ agents
have input $0$, all agents output $1$. 
\end{lemma}

\begin{proof}
  Let $S$ be the subset of agents that have input $0$. Given a scheduler
  $\sigma_e$, consider an adversary $A = (T, \vec{\tau}_T,
\sigma_e)$ such that $T \supseteq S$, $|T| = t$, and $\vec{\tau}_T =
\vec{\sigma}_T$ (so all the malicious agents follow protocol
$\vec{\sigma}$). By definition of $t$-weak secure computation, the 
output of honest agents with adversary $A$ should be one that is possible with a trusted-party adversary of the form $A' = (T, c, h, O)$. However, no
matter what the output $C$ of $c$ is, since $|C| \ge n-t$, there
will be at least $n-2t$ honest agents in $C$, all of them with input
$1$. Since $n \ge 3t$, this suffices to guarantee that all players not
in $T$ output 1. Since malicious agents play $\vec{\sigma}$, they are indistinguishable from honest agents. Thus, if all agents are honest, all agents not in $T$ output $1$. To see that 
agents in $T$ also output $1$ if all players are honest, consider an
adversary $A'' = (T', \vec{\tau}_{T'}, \sigma_e)$ such that $T' \cap T
= \emptyset$, $|T'| = t$ (such a set $T'$ always exists
since $n \ge 3t$), and $\vec{\tau}_{T'} = \vec{\sigma}_{T'}$.
Since  honest agents not in $T \cup T'$ 
(note that $[n] \setminus (T \cup T') \not= \emptyset$)
have the same histories with $A$ and $A''$, they must output the same
value with both adversaries, and so must
output $1$ with adversary $A''$. By definition of $t$-weak secure
computation, since $|T'| = t$, all agents not in $T'$ must output the
same value. Thus, since $T \cap T' = \emptyset$, all agents in $T$
also output $1$ with adversary $A''$. Again, since agents in 
$T'$
 are
indistinguishable from honest agents, this implies that agents in $T$
also output $1$ if all agents are honest. 
\end{proof}

\begin{proof}[Proof of Lemma~\ref{lemma:same-histories}]
\commentout{
With input profile 
$(1,0,1,1)$, 
 all 
agents  
  must
eventually output 1
 even if some 
agent
 does does not send any
messages (since that 
agent
 could be malicious).
Thus,}
By Lemma~\ref{prop:technical-1}, 
 there
exists an integer $N$ such that if 
agents 
 1, 2, and 3 are honest,
with nonzero probability, 
 they will output $1$ with scheduler $\sigma_e^N$
at or before the $N$th time they are scheduled.
Let $H$ be a history where the 
agents 
 use
$\vec{\sigma}$, the scheduler uses $\sigma_e^N$, the
input is 
$(1,0,1,1)$, 
agents  
  1, 2, and 3 are
honest and have
been scheduled 
at most $N$ times and all three have outputted 1.
By the
properties of secure computation, in particular, the secrecy of the
inputs, there must exist a history $H''$ such that $\vec{x}_{H''} = (1,1,0,1)$,
  $H''_1 = H_1$, and $H''_e = H_e$.
  (Note that this means that we can assume, without loss of
generality, that the scheduler uses strategy $\sigma_e^N$.)
If this were not the case and
agent
 1 were malicious in $H''$, then it would know
  that the input profile can't be $(1,1,0,1)$ given histories $H_1$
  and $H_e$. (Recall that we can assume without loss of generality
  that the malicious 
  agents
   can communicate with the
  scheduler.) Similarly, there exists a history $H'$ with
$\vec{x}_{H'} = (0,1,1,1)$  
  such that $H'_2 = H''_2$ and $H'_e =
  H''_e$. The fact that the scheduler has the same local history in
  $H, H'$, and $H''$ and that $H_1 = H''_1$ and $H''_2 = H'_2$ implies
    that $H_{1,2} = H'_{1,2} (= H''_{1,2})$, as desired.  In more detail, since
    $H_2' = H_2''$, agent 2 sends the same messages to and receives
    the same messages from agent 1 in $H'$
    and $H''$, so 1 receives the same messages from and sends the same
    messages to 2 in both $H'$ and $H''$.  Thus, $H_{1,2}' =
    H_{1,2}''$.  A similar argument shows that $H_{1,2} =
    H''_{1,2}$.  
\end{proof}

Now suppose that 
agents 
 have input profile $\vec{x} = (0,0,0,0)$. We show
that there exists a strategy $\tau_3$ for 
agent
 3 
such that if all other 
agents 
 play $\vec{\sigma}$ and the scheduler
plays $\vec{\sigma}_e^N$, then with non-zero probability, 
agents 
 1 and 2 output 1. This suffices to show that 
$f^4$
 cannot be
 $1$-weakly securely computed, since honest 
agents  
  should output 0
  when playing with any trusted-party adversary with at most one
  malicious agent.

  \begin{lemma}
If the agents have input profile $(0,0,0,0)$, then
   there exists a strategy $\tau_3$ for 
agent  3  such that if all other agents 
 play $\vec{\sigma}$ and the scheduler
plays $\vec{\sigma}_e^N$, then with non-zero probability, 
agents  1 and 2 output 1.
\end{lemma}
\begin{proof}  
 Let $H$ and $H'$ be the two histories guaranteed to
exist by
Lemma~\ref{lemma:same-histories}.  The strategy $\tau_3$ for 
agent
 3 consists
of sending 
agent
 1 the messages that 
agent 
  3 sends to 
agent  
   1 in 
$H'$ while
sending 
agent
 2 the messages that 
agent 
  3 sends to 
agent  
   2 in 
$H$.
Suppose that
agent
 1 has the same random bits as in 
$H'$, 
 while
agent
 2 has the same random bits as in 
$H$. 
An easy induction now shows that, in the resulting history, 
agent
 1
will have history 
$H'_1$ 
 and 
agent 
  2 will have history 
 $H_2$
 after
each having been scheduled 
at most 
$N$ times, using the fact that, as shown in
Lemma~\ref{lemma:same-histories}, $H_{12} = H'_{12}$.
Thus, by Lemma~\ref{lemma:same-histories}, 
agents 1 and 2 output 1 in this case.
This contradicts the fact that $\vec{\sigma}$ 1-weakly securely
computes $f^4$, since 
\commentout{
reasoning analogous to that of
Lemma~\ref{prop:technical-1}
}
Lemma~\ref{prop:technical-1} 
 shows that, with input profile
$(0,0,0,0)$, all honest players 
output $0$. 
\end{proof}

It is
straightforward to extend this argument to all $n$ and $t$ such that
$3t \le n 
\le 4t$. Given $n$ and $t$ such that $3t \le n \le 4t$, we divide the
agents into four disjoint sets $S_1$, $S_2$, $S_3$, and $S_4$ such
that $0 < |S_i| \le t$ for all $i \in \{1,2,3\}$ and $0 \le |S_4| \le
t$. Consider a scheduler $\sigma_e^N$ that schedules 
agents
 in
$S_1, S_2$ and $S_3$ cyclically and, right before scheduling an
agent,
it delivers the messages that were sent by the other 
agents
 the last
time they were scheduled. After scheduling each of the 
agents
 in $S_1
\cup S_2 \cup S_3$ $N$ times, it schedules the 
agents
 in $S_4$ as
well. Suppose that $\vec{\sigma}$ is a strategy for $n$ 
agents
 that
$t$-securely computes $f^n$. 

\begin{lemma}\label{lemma:generalization}
There exist $N$ and two (finite) histories $H$ and $H'$ of
$\vec{\sigma}$ where the scheduler uses $\sigma_e^N$, $\vec{x}_H =
(1_{S_1}, 0_{S_2}, 1_{S_3}, 1_{S_4})$,
\commentout{
(which means that 
agents
 in
$S_1$ have input 1, 
agents
 in $S_2$ have input 0, etc.),}
$\vec{x}_{H'} = (\vec{0}_{S_1}, \vec{1}_{S_2}, \vec{1}_{S_3}, \vec{1}_{S_4})$, 
agents
 in $S_1
\cup S_2 \cup S_3$ output 1 in $H$, 
agents
 in $S_4$ are never
scheduled in either $H$ or $H'$, $H_{S_1,S_2} = H'_{S_1, S_2}$ (which
is the conversation between 
the agents
 in $S_1$ and 
 the agents
  in $S_2$) and
$H_e = H'_e$. 
\end{lemma}

\begin{proof} The proof is analogous to the proof of
Lemma~\ref{lemma:same-histories}; the subsets $S_1$, $S_2$, $S_3$, and
$S_4$ play the roles of 
agents
 1, 2, 3, and 4, respectively.
\end{proof}

We now have the tools we need to prove
Theorem~\ref{prop:weak-secure-computation}(b). 
Given $H$ and $H'$ from Lemma~\ref{lemma:generalization},
consider a strategy $\vec{\tau}_{S_3}$ for 
agents
in $S_3$ that
consists of sending 
agents
 in $S_1$ and $S_2$ exactly the same
messages they would send in $H'$ and $H$ respectively. Again, if
agents
 have input $\vec{0}$, with non-zero probability, 
agents 
  in
$S_2$ will eventually have history $H_{S_2}$, and thus will output $1$,
contradicting the assumption that $\vec{\sigma}$ 
$t$-weakly securely computes $f^n$.
This completes the proof of Theorem~\ref{prop:weak-secure-computation}(b). 

The proof of Theorem~\ref{prop:secure-computation}(b)
follows similar lines.  We start with 
an analogue of Lemma~\ref{prop:technical-1} which holds for a larger
range of values of $n$: 

\commentout{
In the case of standard secure computation, the proof of Theorem~\ref{prop:weak-secure-computation} is valid for a larger range of values of $n$. In fact we have the following:

\begin{theorem}\label{thm:secure-computation-lower-bound}
  If 
$t+2 \le n \le 4t$, 
   there exists a domain $D$ and a function $f: D^n \rightarrow D$ that cannot be $t$-securely computed. 
\end{theorem} 


Suppose that $n \ge t+2$ and that $\vec{\sigma}$ $t$-securely computes
$f^n$. If we split the agents into four groups $S_1, S_2, S_3, S_4$
such that $|S_i| = \lceil \frac{n-t}{3} \rceil$ for $i \in \{1,2,3\}$
(and such that $|S_4| \le t$) we have the following: 
}

\begin{lemma}\label{prop:technical-3}
Let $n \ge t+2$ and $\vec{\sigma}$ be a protocol that $t$-securely computes $f^n$. Then, in all histories of $\vec{\sigma}$ in which all agents are honest and at most $(n-t)/2$ agents have input 0, all agents output 1.
\end{lemma}

\begin{proof}
Given any scheduler $\sigma_e$, if all agents are honest, their output should be one that is possible with a trusted-party adversary of the form $A = (\emptyset, c,h,O)$. No matter what the output $C$ of $c$ is, at most $(n-t)/2$ agents in $C$ have input $0$. Since $|C| \ge n-t$, at least half of the agents in $C$ have input $1$ and thus all honest agents output 1.
\end{proof}

\commentout{
\begin{proposition}\label{prop:technical-3}
In all histories of $\vec{\sigma}$ such that agents have input profile $(0_{S_1}, 1_{S_2}, 1_{S_3}, 1_{S_4})$, $(1_{S_1}, 0_{S_2}, 1_{S_3}, 1_{S_4})$, $(1_{S_1}, 1_{S_2}, 0_{S_3}, 1_{S_4})$ or $(1_{S_1}, 1_{S_2}, 1_{S_3}, 0_{S_4})$, all agents output $1$.
\end{proposition}

\begin{proof}
For any scheduler $\sigma_e$, if all players are honest, their output should be one that is possible with a trusted-party adversary of the form $A = (\emptyset, c,h,O)$. No matter what the output $C$ of $c$ is, at most $\lceil \frac{n-t}{3} \rceil$ of agents in $C$ have input $0$. Since $n \ge t+2$ we have that $n-t \ge 2 \lceil \frac{n-t}{3} \rceil$, which means that $C$ has at least as many agents with input $1$ than agents with input $0$. This implies that all honest players output $1$.
\end{proof}
}

\commentout{
The remaining part of the proof 
of Theorem~\ref{thm:secure-computation-lower-bound}
is analogous to that of Theorem~\ref{prop:weak-secure-computation}.
}

We also need the following technical result:
\begin{lemma}\label{prop:inequalities}
If $t+2 \le n \le 4t$ then
\begin{itemize}
  \item[(a)] $n \ge 3\lceil \frac{n-t}{3}\rceil$;
\item[(b)] $\lceil \frac{n-t}{3}\rceil \le \frac{n-t}{2}$.
\end{itemize}
\end{lemma}

\begin{proof}
  If $t + 2 \le 4t$ then $t > 0$. To prove part (a), note that if $t = 1$,
  then $n$ can be only $3$ or $4$. In both cases, the inequality is
satisfied. If $t \ge 2$ then $\lceil \frac{n-t}{3}\rceil \le \lceil
\frac{n-2}{3}\rceil \le \frac{n}{3}$,
from which the desired result immediately follows.
To prove part (b), let $a$ and $b$ be the two positive
integers such that $n-t = 3a + b$ with $1 \le b \le 3$. Then $\lceil
\frac{n-t}{3}\rceil = a+1$ and $\frac{n-t}{2} = \frac{3a+b}{2} = a +  
\frac{a+b}{2}$. Since $n-t \ge 2$, then either $a > 0$ or $b >
1$.  Since $b \ge 1$, in both cases, $a+1 \le a + \frac{a+b}{2}$. 
\end{proof}

Given $n$ and $t$ such that $t+2 \le n \le 4t$, we divide the agents
into four disjoint sets $S_1, S_2, S_3, S_4$ such that $|S_i| = \lceil
\frac{n-t}{3}\rceil$
for $i \le 3$ and $|S_4| \le t$ (which is always possible, by
Lemma~\ref{prop:technical-3}(a)). 
If $n \ge t+2$, then by Lemma~\ref{prop:inequalities}(b),
$\lceil \frac{n-t}{3}\rceil \le  
\frac{n-t}{2}$,
and thus by Lemma~\ref{prop:technical-3}, in all
histories in which all agents are honest and have inputs $(0_{S_1},
1_{S_2}, 1_{S_3}, 1_{S_4})$, $(1_{S_1}, 0_{S_2}, 1_{S_3}, 1_{S_4})$ or
$(1_{S_1}, 1_{S_2}, 0_{S_3}, 1_{S_4})$, all the agents output
$1$. 
Reasoning analogous to that used in the proof of 
Theorem~\ref{prop:weak-secure-computation}(b) then shows that $f^n$
cannot be $t$-securely computed for $t+2 \le n \le 4t$. 

It remains to deal with the case where $n = t+1$. To show that
there exist functions that cannot be $t$-securely computed if $n =
t+1$, we reduce \emph{$t$-resilient weak consensus} to $t$-secure
computation. 
\begin{definition}
A protocol $\vec{\sigma}$ for $n$ agents is a $t$-resilient implementation of weak consensus if the following holds for all adversaries $A = (T, \vec{\tau}_T, \sigma_e)$ with $|T| \le t$ and all histories:
\begin{itemize}
\item[(a)] All agents not in $T$ output the same value.
\item[(b)] If all agents are honest and have the same input $x$, all agents output $x$. 
\end{itemize}
\end{definition}
Lamport~\citeyear{L83} 
showed that if $n \le 3t$ there is no $t$-resilient implementation of
weak consensus. Thus, if there exists a reduction from
$t$-resilient weak consensus to $t$-secure computation for $n > t$,
then there are functions $f:D^n \rightarrow D$ with $n = t+1$
that cannot be $t$-securely computed. 

The reduction proceeds as follows: Consider a function $g^n: \{0,1, \bot\}
\rightarrow \{0,1,\bot\}$ such that $g^n(\bot, \ldots, \bot) = \bot$,
and $g^n(x_1, \ldots, x_n) = x_i$ if $x_i \not =
\bot$ and $x_j = \bot$ for all $j < i$; that is, $g^n$ outputs
the first non-$\bot$ value if there is one, and otherwise outputs
$\bot$. Suppose, by way of contradiction, that $\vec{\sigma}$
$t$-securely computes $g^n$. 
Let $\vec{\tau}$ be identical to $\vec{\sigma}$, 
except hat, whenever agent $i$ would have output $(C,v)$ with
$\sigma_i$,  it outputs $v$ instead if $v \not = \bot$, and otherwise it outputs
$0$. By the properties of $t$-secure computation,
all honest agents output the same value when using $\vec{\tau}$. 
Moreover, if all honest agents have input $0$ or all of them
have input $1$, if $n > t$, then the output of the
secure computation has the form $(C,0)$ or $(C,1)$, respectively. Thus, if
there exists a protocol that $t$-securely computes $g^n$ for $n =
t+1$, then there exists also a $t$-resilient implementation of weak
consensus for $t+1$ agents, contradicting 
Lamport's
result. This proves Theorem~\ref{prop:secure-computation}.

\section{Proof of Theorem~\ref{thm:secure-computation-reduction}}\label{sec:appendix2}

We  prove Theorem~\ref{thm:secure-computation-reduction} only for the
case of $(k,t)$-robustness; the proof in the case of $(k,t)$-strong
robustness is identical. 

A naive construction of $\Gamma_d$ and $\vec{\sigma} + \sigma_d$
proceeds as follows. The set of actions of each player consists of all
possible outputs of a secure computation of $f$ in addition to their
type (more precisely, actions are of the form $(C,z,Q)$ with $C
\subseteq [n]$ and $z \in D$ and $Q \in \{H, R, M\}$, where $H$ stands
for honest, $R$ stands for rational, and $M$ stands for malicious).
If there is  
no subset $S$ of at least $n-k-t$ honest players such that players in
$S$ securely compute $f$, that is, every subset $S$ of
$n-k-t$ players either do not all output the same value or they all output a
value that is not a possible output of a secure computation of $f$,
then rational players get a higher payoff and/or the honest players
get a lower payoff. In $\vec{\sigma} + \sigma_d$, each player sends
its input to the mediator when it is scheduled for the first time. The
mediator waits until it receives the input $x_i$ from a set $C$ of
players with $|C| \ge n - k - t$, then computes $z := f_C(\vec{x})$
and sends $(C,z,H)$ to all players. Players play $(C,z,H)$ when they
receive the message. 

It would seem that any $(k,t)$-robust implementation of $\vec{\sigma}
+ \sigma_d$ also $(k+t)$-strictly securely computes $f$. In fact, any
$(k+t)$-secure computation of $f$ is also a $(k,t)$-robust
implementation of $\vec{\sigma} + \sigma_d$, but not all $(k,t)$-robust
implementations $(k+t)$-securely compute $f$. As we suggested in the
main text, consider a
protocol $\vec{\sigma}_{ACT}$ in which players perform BCG's
$(k+t)$-secure computation protocol and broadcast their
inputs immediately afterwards. It is easy to check that
$\vec{\sigma}_{ACT}$ is a $(k,t)$-robust implementation of
$\vec{\sigma} + \sigma_d$ whenever $n > 4(k+t)$, but that
$\vec{\sigma}_{ACT}$ does not $(k+t)$-securely compute $f$, since it
leaks the honest players' inputs to all other players. 

This shows that it is necessary to somehow encode all information that
malicious players can learn into the set of actions of $\Gamma_d$ in
such a way that they can increase their payoff if they manage to
learn anything about the other players' inputs besides what can be
learned from the output of the computation. The idea for doing this is
that, besides the output, the action of each player should include
a \emph{guess} as to what the input profile $\vec{x}$ is (they can
also guess $\bot$ if they have no guess). If a player $i$ guesses
correctly it receives an additional positive payoff $q_i$, while if it
guesses wrong, its payoff decreses by $1$. The value of $q_i$
should be chosen in such a way that (a) it is never worthwhile deviating
if $i$ is not able to learn anything besides the
output, and (b) it is always worthwhile deviating if $i$ is able to learn
something (otherwise, $\vec{\sigma}_{ACT}$ 
may not $(k+t)$-strictly securely compute $f$ even if it is
$(k,t)$-robust, as in the example above). Given the set $C$ of players
whose inputs are included in the computation, the output $z$ of $f$,
and the input profile $\vec{x}$, let $p_i$ be the probability that a
player $i$ guesses $\vec{x}$ conditional on its own input $x_i$,
$C$, and $z$. Conditions (a) and (b) imply that $p_iq_i +
(1-p_i)(-1) \le 0$ and $p_iq_i + (1-p_i)(-1) \ge 0$ respectively,
which means that $p_iq_i + p_i - 1 = 0$ and thus 
that $q_i = p_i^{-1} - 1$. 

This approach cannot be 
generalized easily to a situation where a coalition of
players may deviate. In this case, a player in the coalition will know
the values of all players in the coalition, not just its own input.  Moreover,
if a player in the coalition plays just like an honest
player, except that it tells the other coalition
members its input, then  this is completely indistinguishable (by
the honeset players) from the
scenario in which that player is honest and the other members of the
coalition were just lucky guessing its input. In addition, a player can lie
about its input if it is easier to guess the input profile with a
different input than its own. Since the payoffs of a Bayesian game
depend only on their actions and their real input profile, it is
always worthwhile for a player to lie about its input if this is the
case. For instance, suppose that $D = \mathbb{F}_2$ and that
$f(\vec{x}) = \prod_{i = 1}^n (1-x_i)$. If $i$ has input $1$ and plays
honestly, then it
learns absolutely nothing about the other players' inputs, since the output
will be $0$ no matter what. However, if $i$ pretends to have input
$0$, $i$ will learn more 
information: if the output is $1$, then all players have input $0$;
otherwise, at least one player has input $1$. In this case, it would 
always be worthwhile for player $i$ to act as if it has input $0$,
regardless of its actual 
input. 
This shows that to compute the probability that the adversary
guesses the inputs correctly, it is critical to know who is malicious
and what inputs the malicious players are pretending to use in the computation.

To deal with the fact that we may not be able to tell which players
are deviating, we require that exactly $k+t$ players must try
to guess a non-$\bot$ value in order to get an additional (positive or
negative) payoff. Moreover, their guesses must be identical. If
honest players always guess $\bot$, this suffices to identify the
coalition of $k+t$ deviating players given their action profile.
Note that this is why we require strict secure computation in
Theorem~\ref{thm:secure-computation-reduction}. If we required only
(standard) secure computation, smaller adversaries wouldn't be able to
get a better payoff, even if they managed to guess the inputs
of everyone else (thus it wouldn't satisfy condition (b)). 
\commentout{
(note
that this is why we require strict secure computation in
Theorem~\ref{thm:secure-computation-reduction}, rather than standard
secure computation).  
}
\commentout{
As we shall see, this will enable us to
determine the probability of them 
guessing the input profile.
}
To 
deal with players lying about their inputs, we require that the action
profile of the 
players encode the inputs used by the players for the computation
(even though these inputs may differ from their actual inputs). The
probability of guessing the input profile is based on the inputs used,
not players' actual inputs.  Note that these values must be
encoded into the action profile without any coalition of $k+t$ players
learning anything about them. This can be done as follows: each player
$i$, in addition to the set $C$, the output $z$ and their guess $b_i$, 
also outputs $n$ values $s_{i1}, \ldots, s_{in}$ such that the values
$s_{i,j}$ for a fixed $j$ encode the value used by $j$ for the
computation (using Shamir's scheme). 

There is one final issue. The definition of $(0, t)$-robustness
is equivalent to that of $t$-immunity, which means that no coalition
of $t$ players can decrease the payoff of other players by
deviating.  In this case, the effect of a coalition of
$t$ players being able to learn something about the inputs should
be to decrease the payoffs of the remaining players,
rather than increasing their own payoff.  To deal with this, we
require players to declare wither they are 
$G$ (good), $R$ (rational), or $M$ (malicious). If a
coalition of $(k+t)$ players tries to guess the inputs of everyone
else, if they all declare $R$, then they get an additional payoff
as described above. Otherwise, everyone gets the negative of that
value. 

We now formalize these ideas.
\commentout{
Given a function $f: D^n \rightarrow D$,
our goal is to design a game
that corresponds to a secure computation of $f$.  That means that if 
the honest players do not securely compute $f$, which 
means that either (a) honest players 
output different values or
some honest player outputs a value that is not a possible
output of a secure computation of $f$, or (b) rational players are
able to deduce information about the honest players' inputs,
then there is a deviation for the rational and malicious players that
results in the rational players getting a higher payoff and/or the
honest players getting a lower payoff.
More formally, given
}
Given
 $f$ 
and integers $k$ and $t$ such that $n > 2k + 2t$, consider the game
$\Gamma^{f,k,t}$, defined as follows. The input profile of the
players is chosen uniformly at random from $D^n$. The set of
actions of each player in $\Gamma^{f,k,t}$ is $\{G, R, M\} \times 2^{[n]} \times D \times
D^n \times (D^n \cup \bot)$, so an action of player $i$ has 
the form $a_i = (Q_i, C_i, z_i, \vec{s}_i, b_i)$, where 
$Q_i \in \{G, R, M\}$,
$C_i \subseteq [n]$,
 $z_i \in D$, 
$x_i \in D$, and $b_i \in D^n \cup \bot$.
Intuitively, 
$Q_i$ denotes if $i$ is good ($G$), rational ($R$), or malicious ($M$),
$(C_i,z_i)$ is $i$'s output in the secure computation of
$f$; $s_{ij}$ is $i$'s share of $j$'s input (this will be made clearer
below), and $b_i$ is $i$'s guess 
of
the (supposedly secret) input, where $b_i = \bot$ if $i$ has no guess.

We next define the utility function.  We take $u_i = u_i^1 +
u_i^2$, where, 
intuitively,
 $u_i^1$ is the utility that $i$ gets if 
honest players either output different values or
some honest player outputs a value that is not  a possible
output of a secure computation of $f$ and $u_i^2$ is the utility that
$i$ gets from guessing the correct input of the other players.
To define $u_i^1$, we first define what it means
for 
 an action profile
$\vec{a}$ to be \emph{secure for an input profile $\vec{x}$}.  This is
the case if 
there exist subsets $C, T \subseteq [n]$ with $|C| \ge n-t-k$ and
$|T| = k+t$, a vector $\vec{v} \in D^{k+t}$,  
and $n$ polynomials $p_1, \ldots, p_n$ of degree $k+t$ 
(where, intuitively, $p_j$ encodes $j$'s input, so $p_j(i)$ is $i$'s
share of $j$'s input)
such that, for each player $j \notin T$, the action $a_j = (Q_j, C_j, z_j,
\vec{s}_j, b_j)$ of player $j$ satisfies (1) $Q_j = G$, (2) $C_j = C$, (3)
$b_j = \bot$, 
(4) $z_j = f(\vec{y})$, where  $\vec{y} =
(\vec{x}/_{(T,\vec{v})})/_{(\overline{C},\vec{0})}$ (5) $p_{j'}(j) =
s_{jj'}$ for all $j' \in [n]$, and (6) $p_j(0) = y_j$.  We say that
\emph{$\vec{a}$ is $(T,\vec{x})$-secure} if these properties hold for
the set $T$. 
Intuitively, $\vec{a}$ is $(T,\vec{x})$-secure if it could be the
output of a $(k+t)$-secure computation of $f$ with input $\vec{x}$,
and the inputs used for the computation---which may differ from the
the actual input profile due to deviating players lying about their
inputs---were shared
correctly between 
the players. 
If $\vec{a}$ is secure for $\vec{x}$, then 
$u_i^1(\vec{a}, \vec{x}) = 0$ for all $i \in [n]$; if $\vec{a}$ is not
secure for $\vec{x}$ and at least one player $i$ played
$R$ (i.e., played an action with $Q_i = R$), then $u_j^1(\vec{a},
\vec{x}) = 1$ for all players $j$; 
otherwise, $u_j^1(\vec{a}, \vec{x}) = -1$ for all players $j$.

If $\vec{a}$ is
not 
 secure for $\vec{x}$, then $u^2_i(\vec{a}, \vec{x}) = 
0$.
If $\vec{a}$ is secure for $\vec{x}$,
let $K$ be the subset of players that did not play $G$.
Note that if $\vec{a}$ is secure for $\vec{x}$, then $|K| \le
k+t$.  If $|K| < k+t$ or not all players in $K$ guess the same value $b$
(i.e. not all players in $K$ have the same value $b$ as the last
component of their actions), then 
$u^2_i(\vec{a}, \vec{x}) =  
0$ for all $i$. 
\commentout{
Otherwise, let $p$ be the probability that a 
vector $\vec{w}$ sampled uniformly from $D^n$ is equal to $\vec{x}$,
conditional on $\vec{w}_K = \vec{y}_K$ and 
$f_C(\vec{w}) = z$ (note that $\vec{y}$ and $C$ are uniquely determined by $\vec{a}$ if $\vec{a}$ is secure for $\vec{x}$), and let $b$ be the common guess of players in $K$.
}
Otherwise, let $b$ be the common guess of players in $K$ and let $p$
be the probability that a  
vector $\vec{w}$ sampled uniformly from $D^n$ is equal to $\vec{x}$,
conditional on $\vec{w}_K = \vec{y}_K$ and 
$f_C(\vec{w}) = z$. Note that if $\vec{a}$ is $(T,\vec{x})$-secure for
some $T$, then $C$ is uniquely determined by $\vec{a}$, and
if, in addition, $n > 2(t+k)$, then $\vec{y}$ is also uniquely determined by
the shares $\vec{s}_i$ of players $i \not \in T$. 
 If $b = \vec{\bot}$, then $u^2_i(\vec{a}, \vec{x}) = 
0$ for all $i$.
If at least one player $i \in K$ played $R$ in its action, then, if $b
= \vec{x}$, $u^2_i(\vec{a}, \vec{x}) = p^{-1} - 1$ for all
$i \in K$; otherwise, $u^2_i(\vec{a}, \vec{x}) = -1$ for all $i \in K$.
On the other hand, if no player $i \in K$ played $R$ in its action,
then, if $b = \vec{x}$,  $u^2_i(\vec{a}, \vec{x}) = 1 - p^{-1}$ for
all $i \not \in K$; otherwise, $u^2_i(\vec{a}, \vec{x}) = 1$ for
all $i \not \in K$. Note 
that since $p(p^{-1} - 1) - (1 - p) = 0$, the payoffs $u_i^2$ are
designed in such a way that the adversary can, in expectation, either
increase its payoff (if there are any rational players) or decrease
the payoff of everyone else (if there are no rational players) if it
can guess the inputs of honest players with higher probability than
$p$ (which is the probability of guessing the honest players' inputs
if the adversary knew nothing but the output of the function and its
own strategy and inputs).

Consider the following strategy $\vec{\sigma} + \sigma_d$ for
$\Gamma_d^{f,k,t}$. According to $\sigma_i$, player $i$
sends its input to the mediator at the beginning of the game. If $i$
receives a message $msg$ from the mediator, it plays $msg$ in the
underlying game. According to  $\sigma_d$,
the mediator waits until there exists a set 
$C \subseteq [n]$ with $|C| \ge n-t-k$ such that it has received exactly
one message from each player $i \in C$ and each of these messages
consists of a value $y_i \in D$. The mediator computes
$n$ polynomials $p_1, \ldots, p_n \in D[X]$ of degree $k+t$ whose
non-constant coefficients are chosen uniformly at random and such
that $p_i(0) = y_i$ if $i \in C$ and $p_i(0) = 0$ otherwise;
it then 
computes $z := f(p_1(0), \ldots, p_n(0))$ and sends $(G,C,
z, p_1(i), \ldots, p_n(i), \bot)$ to each player $i$. 

\begin{proposition}\label{prop:equilibrium}
  $\vec{\sigma} + \sigma_d$ is 
  $(k,t)$-robust and the equilibrium payoff is $0$.
\end{proposition}

\begin{proof}
Let $u_i(\vec{\sigma}, A)$ be the expected payoff of player $i$ when
playing $\vec{\sigma}$ with adversary $A$. It follows by construction
that $u_i^1(\vec{\sigma} + \sigma_d, A) = 0$ for all adversaries $A =  
 (T, \vec{\tau}_T)$ of size at most 
 $k+t$
  since, no matter what $T$ is,
the output profile $\vec{a}$ is $(T,\vec{x})$-secure for all input profiles $\vec{x}$.

 Thus, $\vec{\sigma} + \sigma_d$ is not 
 $(k,t)$-robust if and
 only if  there exists an adversary $A = (T, \vec{\tau}_T)$ with $|T|
  \le k+t$ and an 
 input profile $\vec{x}_T$ such that, in expectation, (a) $u_i^2(\vec{\sigma} + \sigma_d,
 A, \vec{x}_T) > 0$ for 
 some 
 $i \in T$, or (b) $u_i^2(\vec{\sigma} + \sigma_d,
 A, \vec{x}_T) < 0$ for all $i \not \in K$.  The definition of $u_i^2$
guarantees that, in both cases, the adversary must consist
 of exactly $k+t$ players and these players must all play a non-$G$
 action. Moreover, these players must guess the input of honest
 players with a probability higher than they could guess it by just
 knowing the output of 
 the function, their strategy, and their inputs. However, the
 construction of $\vec{\sigma} + \sigma_d$ guarantees that don't have any extra
 information
 (note that $C$ depends only on the adversary, and that the
  adversary does not get any information about the input of the honest
 players besides the value of $z$, since the polynomials $p_i$ are all
  of degree $k+t$).  
\end{proof}

We next show that if there exists a $(k,t)$-robust strategy that implements $\vec{\sigma} + \sigma_d$, then that strategy also $(k+t)$-securely computes $f$.
We first need the following lemma.

\begin{lemma}\label{lemma:equiv1}
If $\vec{\sigma}_{ACT}$ is a $(k,t)$-robust strategy that implements
$\vec{\sigma} + \sigma_d$, then for all adversaries $A = (T,
\vec{\tau}_T,\sigma_e)$ with $|T| = k+t$, all inputs $\vec{x}$, and all
histories 
$H$
 of $\vec{\sigma}_{ACT}$ with adversary $A$ and
input $\vec{x}$, the action profile $\vec{a}$ played in 
$H$
 is
$(T, \vec{x})$-secure. 
\end{lemma}

\begin{proof}
Suppose that $k > 0$. If there exists an input $\vec{x}$ and an
adversary $A = (T, \vec{\tau}_T, \sigma_e)$
with $|T| = k+t$ 
 such that, for some history
$H$, 
 the action profile $\vec{a}$ played in 
$H$
  is not $(T, \vec{x})$-secure,
consider the adversary $A' = (T, \vec{\tau}'_T, \sigma_e)$ where
$\vec{\tau}'_T$ is identical to $\vec{\tau}_T$, except that if a
player $i \in T$ plays an action $a$
with $\tau_i$, then $i$ instead plays
  $(R, \emptyset, 0, 0, \bot)$ with $\tau'_i$.
Thus,
 if 
an
  action profile $\vec{a}'$
 played in some history 
$H'$ 
 when $\vec{\sigma}_{ACT}$ is run with
adversary $A'$  
is 
$(T', \vec{x})$-secure, then $T\subseteq T'$
(note that for an action profile $\vec{a}$ to be $(T',
\vec{x})$-secure, we require that all players not in $T'$ play $G$ in
the first component, and none of the players in $T$ plays $G$) 
 and, since $|T| = k+t$, $T = T'$.
 Since the histories generated by playing with adversaries $A$ and
 $A'$ are indistinguishable by honest players,   
 if there exists a history  
$H$
 with adversary $A$
 and input $\vec{x}$ such that the action $\vec{a}$ played in 
$H$ 
  is
not
  $(T, \vec{x})$-secure, 
then the resulting action profile $\vec{a}'$ of playing
$\vec{\sigma}_{ACT}$ with adversary $A'$ and input $\vec{x}$ in which
all players use the same randomization as in 
$H$ 
 is not 
$(T,\vec{x})$-secure, 
and all players in $T$
would get a payoff of 1 rather than 0. It follows that
$\vec{\sigma}_{ACT}$ is not $(k,t)$-robust. If $k = 0$, the argument is
analogous, except that players in $T$ play $(M, \emptyset, 0, 0, \bot)$
rather than $(R, \emptyset, 0, 0, \bot)$
\end{proof}

\commentout{ 
The desired result now follows immediately:

\begin{proposition}\label{prop:equivalence}
If $g$ is the function that, given $(Q, C, v, \vec{s}, b)$ with $Q \in \{G,R,M\}$, $C
\subseteq[n]$, $v \in D$,  $\vec{s} \in D^n$, and $b \in D^n \cup \bot$,
returns the action of outputting $(C,v)$, and 
$\vec{\sigma}_{ACT}$ is a $(k,t)$-robust strategy
that implements 
$\vec{\sigma} + \sigma_d$, then $g(\vec{\sigma}_{ACT})$
$(k+t)$-strictly securely 
computes $f$. 
\end{proposition}
}

To complete the proof of
Theorem~\ref{thm:secure-computation-reduction}, we must 
show that if $\vec{\sigma}_{ACT}$ is a $(k,t)$-robust implementation
of $\vec{\sigma} + \sigma_d$, the output of an adversary $A = (T,
\vec{\tau}_T, \sigma_e)$ with $|T| = k+t$ is just a (randomized)
function of its input $\vec{x}_T$ and the output $v$ of the function.  
To do this, we need the following lemma:
\begin{lemma}\label{lemma:probability}
Consider two random variables $X$ and $Y$ that take values on  countable spaces $S_1$ and $S_2$ respectively. Then, $\Pr[X = x \mid Y = y] = \Pr[X = x \mid Y = y']$ for all $x \in S_1$ and $y,y' \in S_2$ if and only if $X$ and $Y$ are independent.
\end{lemma}
\begin{proof}
If $\Pr[X = x \mid Y = y]$ does not depend on $y$, there exists a
constant $\lambda_x$ such that $\Pr[X = x \mid Y = y] = \lambda_x$ for
all $y \in S$. Then, since $\Pr[X = x \mid Y = y] = \frac{\Pr[X = x, Y
    = y]}{\Pr[Y = y]}$, it follows that $\Pr[X = x, Y = y] = \lambda_x
\Pr[Y = y]$. Therefore, $\sum_{y \in S_2} \Pr[X = x, Y = y] = \sum_{y
  \in S_2} \Pr[Y = y]$, which gives that $\lambda_x = \Pr[X = x]$, as
  desired. The converse is straightforward. 
\end{proof}

We can now complete the proof of Theorem~\ref{thm:secure-computation-reduction}.
Suppose that $k > 0$. Recall that if all players $i \in T$ set $b_i$ to some input
$\vec{x}$, they get a payoff of $p_{\vec{x}}^{-1}-1$, where $p_x\in
    [0,1]$  if the input
profile is indeed $\vec{x}$, and otherwise they get $-1$.
Given a history
$H_T$
 in which the adversary has
input $\vec{x}_T$ and honest players output $v$, let
$p_v^{H_T}(\vec{x})$
 be the probability that the input profile
is $\vec{x}$ conditional on $v$ and 
$H_T$. If
$p_v^{H_T}(\vec{x}) > p_{\vec{x}}$, then
$p_v^{H_T}(\vec{x})(p_{\vec{x}}^{-1} - 1) + (-1)(1 -
p_v^{H_T}(\vec{x})) > 0$, which means that taking $b_i = \vec{x}$
is strictly better than taking $b_i = \bot$ for each of the players in $T$,
contradicting the assumption that $\vec{\sigma}_{ACT}$ is 
$(k,t)$-robust. 
Thus, $p_v^{H_T}(\vec{x}) \le p_{\vec{x}}$ for all $\vec{x}$.
Since both $\sum_{\vec{x}} 
p_v^{H_T}(\vec{x})$ and $\sum_{\vec{x}} p_{\vec{x}}$ are
1, it must be the case  that $p_v^{H_T}(\vec{x}) = 
p_{\vec{x}}$ for all $\vec{x}$. 
This shows that for every history $\vec{h}_T$ of the adversary, the
distribution of possible inputs of honest players conditional on
$\vec{h}_T$ depends only on their inputs and what honest players
output. By Lemma~\ref{lemma:probability}, this implies that
the input of honest players and the history of the adversary are
independent (given the input $\vec{x}_T$ of the adversary and the
output $v$ of honest players), and thus, again by
Lemma~\ref{lemma:probability}, it follows that the distribution of
possible histories of the adversary depends only on $\vec{x}_T$
and $v$. This shows that every possible output function of the
adversary can be simplified to a function that has as inputs only
$\vec{x}_T$ and $v$, as desired. The argument for $k = 0$ is analogous,
\commentout{
considering that in these case it is the input of honest players the
one that varies depending if the adversary guesses the input or not. 
}
except that in this case, if the distribution of possible histories of
the adversary is not independent of the inputs of the honest
players, the adversary decreases the payoffs of the honest players,
rather than increasing the payoffs of the deviating players.
This completes the proof of Theorem~\ref{thm:secure-computation-reduction}.

\commentout{
As an immediate corollary, we get the desired lower bound for
implementing mediators.

\begin{corollary}\label{thm:main}
If 
$3k + 3t \le n \le 4k + 4t$ 
 there exists a 
$(k,t)$-robust (resp., strongly robust)
 protocol $\vec{\sigma} + \sigma_d$ for $n$ players and a
mediator such that there is no
$(k,t)$-robust (resp., strongly robust)
protocol $\vec{\sigma}_{ACT}$ that implements $\vec{\sigma} +
\sigma_d$. 
\end{corollary}
}

Note that a $(k,t)$-robust implementation of $\vec{\sigma} +
\sigma_{ACT}$ may not necessarily (non-strictly) $(k+t)$-securely
compute $f$, since if the adversary consists of
fewer than $k+t$ malicious players, the malicious players might be
able to deduce 
information about the honest players' inputs without being able to
take advantage of it (recall that a subset $K$ consisting of $k+t$
players must all guess the same value for $u_i^2(\vec{a}, \vec{x})$ to be
non-zero). However,  
a
 small variation in the construction
of $u_i^2$ in $\Gamma_d^{f,k,t}$ allows us to construct a game 
 $\Gamma_d^{f,k}$ such that any 
strongly  
 $(k,0)$-robust implementation of the
strategy used in Proposition~\ref{prop:equilibrium} also 
$k$-securely computes $f$, 
so secure
computation can be reduced to implementing strategies for certain
mediator games. The idea is that instead of
requiring the subset $K$ of players who do not play $G$ to have size
exactly $k$, we only require it to have size at most
$k$.
This modification of $u_i^2$ leads to some of the
problems discussed at the beginning of this section, namely, that if
some rational players act like honest players, except that
they share their inputs with other rational players, the latter players
might be able 
to guess the input profile and get a strictly positive expected payoff.
This scenario is indistinguishable
from one in which the players who shared their input are actually
honest and rational players are just lucky. 
To deal with this issue, 
we
further modify  the payoffs in $\Gamma^{f,k}$
so that
 if
  the players in $K$ guess the inputs correctly, 
    then everyone else gets a huge negative payoff (rather than 0, as in
the original construction). We can show that if this payoff is
sufficiently small (e.g., $-n$ times the winnings), then if there
exists a strategy in which rational players get a positive payoff from
$u^2$, then there exists a strategy in which rational players get a
positive payoff from $u^2$ and they all guess the same value in every
possible history (if the negative payoffs are small enough, 
rational players not guessing any value gives a negative total payoff
for rational players, 
 regardless if some of them guess the correct value). 

 Note that this modification works only for strong $(k,t)$-robustness, since
 if we require only $(k,t)$-robustness, a rational player
may decrease its own payoff if that helps other rational players, even
if the total gain from doing so is negative. 
This is enough to show that
the strategy used in Proposition~\ref{prop:equilibrium} is strongly $(k,0)$-robust with these payoffs.

\section{Proof of Theorem~\ref{thm:reduction-to-weak2}}\label{sec:appendix3}

Consider the game $\Gamma^{k,t}$ in which the set of actions of each player is $\{G,R\} \times \{0,1\}$. Given an action profile $\vec{a}$, in which each player $i$ plays $a_i = (Q_i, y_i)$ with $Q_i \in \{G,R\}$ and $y_i \in \{0,1\}$, let $T$ be the subset of players $i$ such that $Q_i = R$. If $|T| > k+t$, if $k = 0$ all players get a payoff of -1, otherwise all players get a payoff of 1. If $|T| = t+k$ and there exist two players $i,j \not \in T$ such that $y_i \not = y_j$, if $k = 0$ all players get a payoff of -1, otherwise all players get a payoff of 1. In all remaining cases, all players get a payoff of 0. Let $g$ be the function such that $g(Q, y) = y$.

Consider the following protocol $\vec{\sigma} + \sigma_d$ for $n$
players and a mediator. With $\sigma_i$, each player $i$ sends the
mediator its input $x_i$ the first time it is scheduled. The mediator
waits until receiving a message containing either $0$ or $1$, and
sends that value $y$ to all players. The players play $(G, y)$
whenever they receive $y$ from the mediator. Clearly, this strategy is
$(k,t)$-robust (resp., strongly $(k,t)$-robust), since the only way
that players get a payoff other than 0 with an adversary of size 
at most $k+t$ is if two honest players output different values,
but they all receive the same value from the mediator. Suppose a
strategy $\vec{\sigma}_{ACT}$ is a $(k,t)$-robust (resp., strongly
$(k,t)$-robust) implementation of $\vec{\sigma} + \sigma_d$. We show
next that (a) for all adversaries $A = (T, \vec{\tau}_T, \sigma_e)$
with $|T| \le k+t$, all honest players play the same value $y_i$, and
(b) if all players are honest and have the same input $x$, then
they output $x$.  

Property (b) follows trivially from the fact that $\vec{\sigma}_{ACT}$ implements $\vec{\sigma} + \sigma_d$: if all players are honest and have the same input $x$, the value received by the mediator in $\vec{\sigma} + \sigma_d$ is guaranteed to be $x$, and thus, in $\vec{\sigma} + \sigma_d$, all honest players play $(G, x)$. 

To prove (a),
suppose that there exists an adversary $A = (T, \vec{\tau},
\sigma_e)$ with $|T|
\le k+t$ such that, in some history $H$ of
$\vec{\sigma}_{ACT}$ with $A$, there exist two players $i,j \not \in T$ that
play $(Q_i, y_i)$ and $(Q_j, y_j)$, respectively, with $y_i \not =
y_j$, $Q_i = R$, or $Q_j=R$. 
Consider an adversary $A' = (T', \vec{\tau}_{T'}, \sigma_e)$
such that 
$|T'| = k+t$, $T \subseteq T'$, and $i,j \not \in T'$ (we
know that such a  subset $T'$ exists, since $n > t+k+1$), and such that
players in $T$ act as in $\vec{\tau}_T$ and players in $T' - T$ 
act like honest players, except that all of them play $(R,
0)$.
Since histories generated by playing with $A$ and $A'$ are
indistinguished by honest players, 
there exists a history $H'$ in $\vec{\sigma}_{ACT}$ with adversary
$A'$ in which all honest players send and receive the same messages,
and perform the same actions.
If $Q_i = R$ in $H$, then there are $k+t+1$ players that
play $R$ in $H'$: the $k+t$ players in $T'$ and $i$.
Thus, all players get a payoff of 1 if $k > 0$, contradicting the
assumption that  $\vec{\sigma}_{ACT}$ is $(k,t)$-resilient, or 
all players get a payoff of $-1$ if $k=0$,
contradicting the assumption that $\vec{\sigma}_{ACT}$ is $t$-immune.
The same argument shows that $Q_i = G$ in $H$ and $H'$ and, indeed, 
that all honest players must play $G$ in $H$ and $H'$.  Now if $q_i \ne q_j$ in
$H$, then $q_i \ne q_j$ in $H'$, so (since all honest players play
$G$, so exactly $k+t$ players in $H'$ play $R$), again, all players in $H'$
get a payoff of 1 if $k> 0$ and a payoff of $-1$ if $k=0$, so we again
get the same contradiction as before.
%

\bibliographystyle{ACM-Reference-Format}
\bibliography{joe,game1}

\begin{thebibliography}{}

\bibitem[\protect\citeauthoryear{Abraham, Dolev, Geffner, and Halpern}{Abraham
  et~al.}{2019}]{ADGH19}
Abraham, I., D.~Dolev, I.~Geffner, and J.~Y. Halpern (2019).
\newblock Implementing mediators with asynchronous cheap talk.
\newblock In {\em Proc.~38th ACM Symposium on Principles of Distributed
  Computing}, pp.\  501--510.

\bibitem[\protect\citeauthoryear{Abraham, Dolev, Gonen, and Halpern}{Abraham
  et~al.}{2006}]{ADGH06}
Abraham, I., D.~Dolev, R.~Gonen, and J.~Y. Halpern (2006).
\newblock Distributed computing meets game theory: robust mechanisms for
  rational secret sharing and multiparty computation.
\newblock In {\em Proc.~25th ACM Symposium on Principles of Distributed
  Computing}, pp.\  53--62.

\bibitem[\protect\citeauthoryear{Abraham, Dolev, and Halpern}{Abraham
  et~al.}{2008}]{ADH07}
Abraham, I., D.~Dolev, and J.~Y. Halpern (2008).
\newblock Lower bounds on implementing robust and resilient mediators.
\newblock In {\em Fifth Theory of Cryptography Conference}, pp.\  302--319.

\bibitem[\protect\citeauthoryear{Abraham, Dolev, and Stern}{Abraham
  et~al.}{2020}]{ADS20}
Abraham, I., D.~Dolev, and G.~Stern (2020).
\newblock Revisiting asynchronous fault tolerant computation with optimal
  resilience.
\newblock In {\em Proc.~39th ACM Symposium on Principles of Distributed
  Computing}, pp.\  139--148.

\bibitem[\protect\citeauthoryear{Ben{-}Or, Canetti, and Goldreich}{Ben{-}Or
  et~al.}{1993}]{BCG93}
Ben{-}Or, M., R.~Canetti, and O.~Goldreich (1993).
\newblock Asynchronous secure computation.
\newblock In {\em STOC '93: Proceedings of the 25 Annual ACM Symposium on
  Theory of Computing}, New York, NY, USA, pp.\  52--61. ACM Press.

\bibitem[\protect\citeauthoryear{Ben{-}Or, Goldwasser, and Wigderson}{Ben{-}Or
  et~al.}{1988}]{BGW88}
Ben{-}Or, M., S.~Goldwasser, and A.~Wigderson (1988).
\newblock Completeness theorems for non-cryptographic fault-tolerant
  distributed computation.
\newblock In {\em Proc.~20th ACM Symp.~Theory of Computing}, pp.\  1--10.

\bibitem[\protect\citeauthoryear{Ben{-}Or, Kelmer, and Rabin}{Ben{-}Or
  et~al.}{1994}]{BKR94}
Ben{-}Or, M., B.~Kelmer, and T.~Rabin (1994).
\newblock Asynchronous secure computations with optimal resilience (extended
  abstract).
\newblock In {\em Proc.~13th ACM Symp.~Principles of Distributed Computing},
  pp.\  183--192.

\bibitem[\protect\citeauthoryear{Bracha}{Bracha}{1984}]{Br84}
Bracha, G. (1984).
\newblock An asynchronous $[(n-1)/3]$-resilient consensus protocol.
\newblock In {\em Proc.~3rd ACM Symposium on Principles of Distributed
  Computing}, pp.\  154--162.

\bibitem[\protect\citeauthoryear{Canetti}{Canetti}{1996}]{canetti96studies}
Canetti, R. (1996).
\newblock {\em Studies in Secure Multiparty Computation and Applications}.
\newblock Ph.D. thesis, Technion.

\bibitem[\protect\citeauthoryear{Lamport}{Lamport}{1983}]{L83}
Lamport, L. (1983).
\newblock The weak {B}yzantine generals problem.
\newblock {\em J. ACM\/}~{\em 30\/}(3), 668--676.

\bibitem[\protect\citeauthoryear{Shamir}{Shamir}{1979}]{shamir}
Shamir, A. (1979).
\newblock How to share a secret.
\newblock {\em Communications of the ACM\/}~{\em 22}, 612--613.

\end{thebibliography}
\end{document}